\documentclass[11pt]{article}
\usepackage[hmargin=1in,vmargin=1in]{geometry}
\usepackage{hchang}

\usepackage{comment}
\usepackage[hyperref]{ntheorem}
\usepackage{cleveref}

\theoremseparator{.}
\theorembodyfont{\slshape}
\newtheorem{theorem}{Theorem}[section]
\newtheorem{lemma}[theorem]{Lemma}
\newtheorem{claim}[theorem]{Claim}

\newtheorem{question}[theorem]{Question}
\newtheorem{remark}[theorem]{Remark}
\newtheorem{definition}[theorem]{Definition}

\def\eps{\e}
\def\ecd{L_{G, M}(S)}
\def\bdry{\partial\!}
\newcommand{\patt}{\mathbf{p}}
\newcommand{\Patt}{\mathbf{P}}

\newcommand{\ecc}{\mathsf{ecc}}
\def\reals{\mathbb{R}}
\newcommand{\Edist}[2]{\left\lVert #1,#2\right \rVert}
\newcommand{\rep}{\mathsf{rep}}

\title{Computing Diameter\texttt{+}1 in Truly Subquadratic Time for Unit-Disk Graphs\footnote{A conference version has appeared in the 40th International Symposium on Computational Geometry (SoCG 2024). }}

\author{%
Hsien-Chih Chang%
\thanks{Department of Computer Science, Dartmouth College. Email: {\tt hsien-chih.chang@dartmouth.edu}.}  
\and 
Jie Gao%
\thanks{Department of Computer Science, Rutgers University. Email: {\tt jg1555@rutgers.edu}. Supported partially by DMS-2220271, DMS-2311064, IIS-2229876, CCF-2208663, CCF-2118953.}  
\and 
Hung Le%
\thanks{Manning CICS, UMass Amherst. Email: {\tt hungle@cs.umass.edu}. Supported by the NSF CAREER Award No. CCF-2237288, an NSF Grant No. CCF-2121952, and a Google Scholar Research Award.}  
}

\begin{document}

\maketitle
\begin{abstract}
    Finding the diameter of a graph in general cannot be done in truly subquadratic assuming the Strong Exponential Time Hypothesis (SETH), even when the underlying graph is unweighted and sparse. When restricting to concrete classes of graphs and assuming SETH, planar graphs and minor-free graphs admit truly subquadratic algorithms, while geometric intersection graphs of unit balls, congruent equilateral triangles, and unit segments do not. 
    \emph{Unit-disk graphs}  is one of the major open cases where the complexity of diameter computation remains unknown.
    More generally, it is conjectured that a truly-subquadratic time algorithm exists for pseudo-disk graphs where each pair of objects has at most two intersections on the boundary.
    
    In this paper, we show a truly-subquadratic algorithm of running time $\Tilde{O}(n^{2-1/18})$, for finding the diameter in a unit-disk graph, whose output differs from the optimal solution by at most 1.
    This is the first algorithm that provides an additive guarantee in distortion, independent of the size or the diameter of the graph.
    Our algorithm requires two important technical elements. 
    First, we show that for the intersection graph of pseudo-disks, the \emph{graph VC-dimension} --- either of $k$-hop balls or the distance encoding vectors --- is 4.  This contrasts to the VC dimension of the pseudo-disks themselves as geometric ranges (which is known to be 3).
    Second, we introduce a \emph{clique-based $r$-clustering} for geometric intersection graphs, which is an analog of the $r$-division construction for planar graphs. We also showcase the new techniques by establishing new results for distance oracles for unit-disk graphs with subquadratic storage and $O(1)$ query time. 
    The results naturally extend to unit $L_1$- or $L_\infty$-disks and fat pseudo-disks of similar size. Last, if the pseudo-disks additionally have bounded ply, we have a truly-subquadratic algorithm to find the \emph{exact} diameter.
\end{abstract}

{

\section{Introduction}

Given a set $F$ of $n$ objects in the $d$-dimensional Euclidean space $\reals^d$, the \EMPH{geometric intersection graph $F^\times$} has vertices representing the objects in $F$ and edges representing two overlapping objects. When the objects $F$ are disks of radius $1$, the intersection graph is called the \EMPH{unit-disk graph}, where the vertices are centers of the disks in $F$ and two vertices are connected if and only if their distance is no more than $2$.
Unit-disk graphs have been widely used to model wireless communication. 
It is also an interesting family of graphs that admit approximation schemes for many graph optimization problems~\cite{Hunt1998-ns, Nieberg2005-hv}.

Geometric intersection graphs, unlike planar graphs, can be dense. 
But such graphs can be implicitly represented by storing only the set of objects, and the existence of an edge in the graph can often be verified by directly examining the two corresponding objects. 
Thus many algorithms on geometric intersection graphs avoid computing the set of edges explicitly.  
For example, single-source shortest paths in (unweighted) unit-disk graphs can be done in time $O(n\log n)$~\cite{Efrat2001-hm,Cabello2015-vo,Chan2016-sy}, even though the graph may have $\Theta(n^2)$ many edges. 
All-pairs shortest paths can be solved in near-quadratic time for several geometric intersection graphs, including disks, axis-parallel segments, fat triangles in the plane, and boxes in constant dimensional spaces~\cite{Chan2017-oa}. 

In this paper, we examine two distance-related problems, namely, the graph diameter problem and the distance oracle problem for geometric intersection graphs, in particular for unit-disk graphs.  See \Cref{sec:related} for a discussion of prior work on this problem.
A fundamental problem in this area is to determine whether 
\textsc{Diameter} problem can be solved in truly-subquadratic time for geometric intersection graphs. 
This is answered negatively for many types of geometric intersection graphs~\cite{Bringmann2022-me} using a reduction from the Orthogonal Vector Conjecture~\cite{williams2005new} (which is implied by SETH): 
Deciding if diameter is at most $3$ for unit segments in $\reals^2$ and congruent equilateral triangles (with rotation) in $\reals^2$; if diameter is at most $2$ for axis-parallel hypercubes in $\reals^{12}$; and if diameter is at most $k$ for unit balls in $\reals^3$, axis-parallel unit cubes in $\reals^3$ and axis-parallel line segments in $\reals^2$. 
On the positive side, one can decide in $O(n\log n)$ time whether graph diameter is \emph{at most two} for unit-square graphs in $\reals^2$.
However, for unit-disk graphs, arguably the most basic intersection graphs, the complexity of \textsc{Diameter} problem remains wide open. 

\begin{question}
    Can we compute the diameter of unit-disk graphs in truly-subquadratic time?
    \label{quest:uit-disk}
\end{question}

Currently, there is no strong evidence that the answer of Question \ref{quest:uit-disk} is positive or negative. 
As we mentioned above, \textsc{Diameter} for unit-ball graphs in dimension at least $3$ does not have a truly-subquadratic time algorithm unless the Orthogonal Vector Conjecture is false. 
On the other hand, dimension $2$ is fundamentally different from dimension 3 or above, and there are problems that are hard for dimension 3 or above but become much easier in $\mathbb{R}^2$~\cite{Bringmann2022-me}. 

Given the lack of progress on Question \ref{quest:uit-disk}, it is natural to consider approximation algorithms. 
When edges in unit-disk graphs are given their Euclidean distances as weights, finding $(1+\eps)$-approximation of the graph diameter takes $\tilde{O}(n^{3/2})$ time~\cite{Gao2005-kx}; this is later improved to near-linear time~\cite{Chan2019-vw}.  Their approach could be modified to handle unweighted unit-disk graphs to get a \emph{hybrid $(1+\eps, +(4+2\eps))$-approximation} algorithm for \textsc{Diameter}, meaning that the returned approximate diameter is at most $(1+\eps)D + (4+2\eps)$ where $D$ is the true diameter. 
The additional additive error is because when the edges are weighted, for a dense set of disks (e.g., forming cliques of arbitrary size) we can use a subset of disks of density $O(1/\eps^2)$ to obtain a $(1+\eps)$-multiplicative distance approximation; this is no longer true in the unweighted setting --- even removing one disk can potentially introduce a constant additive error to the diameter.  While these results indicate that being on the Euclidean plane helps, stronger evidence supporting a positive answer for Question \ref{quest:uit-disk} would be a \emph{$+\beta$-additive approximation}, where the returned diameter lies in between $D$ and  $D + \beta$. 

\begin{question}
\label{quest:uit-disk-additive}
    Can we compute $+\beta$-approximation of the diameter of (unweighted) unit-disk graphs for some constant $\beta$ in truly-subquadratic time?
\end{question}

A much more general and harder problem is to compute the diameter for the intersection graphs of pseudo-disks~\cite{Bringmann2022-me}. 
Not surprisingly, we are very far from having the answer, given that the unit-disk case remains wide open (\Cref{quest:uit-disk}). 
Unlike the unit-disk graphs, to the best of our knowledge, there are no known \emph{non-trivial approximation} of the diameter in truly-subquadratic time, even for pseudo-disks with constant complexity. In this work, we consider the possibility of obtaining a purely additive approximation of diameter for the intersection graphs of pseudo-disks with constant complexity that have reasonable shapes.  Specifically, we assume that the pseudo-disks are \emph{fat objects} that have \emph{constant complexity} and are similar in size --- those that can be sandwiched between two disks of the same center of radius $r$ and $R$, where $r\leq R$ being two universal constants. 
These objects generalize unit disks and include other objects like unit $L_{p}$-disks, as well as same-size constant-sided convex polygons.

\begin{question}
\label{quest:pseud-disk-additive}
    Can we compute $+\beta$-approximation of the diameter of (unweighted) intersection graphs of similar-size pseudo-disks with constant complexity for {some constant $\beta$} in truly-subquadratic time? 
\end{question}

One source of difficulty in computing diameter in truly-subquadratic time of geometric intersection graphs is that the explicit representation of the intersection graphs could have $\Theta(n^2)$ edges. 
This naturally raises the question of obtaining such an algorithm for \emph{sparse intersection graphs}, where the number of edges is $O(n^{2-\delta})$ for some constant $\delta>0$. 
The answer to this question also remains open. 
A significant progress toward answering this question would be the case of \emph{constant ply}. 
A set of objects is said to have \EMPH{ply $k$} if every point in the space can stab at most $k$ objects in the set. 

\begin{question}
\label{quest:pseud-disk-exact-ply}
    Can we compute the \emph{exact diameter} of (unweighted) intersection graphs of similar-size pseudo-disks with constant complexity and ply in truly-subquadratic time? 
\end{question}

A positive answer to \Cref{quest:pseud-disk-exact-ply} also provides evidence for a positive answer to \Cref{quest:uit-disk}, as unit-disk graphs of constant ply is a special case of similar-size pseudo-disks with constant complexity and ply.

\subsection{Main Results}

In this work, we resolve Questions~\ref{quest:uit-disk-additive}, \ref{quest:pseud-disk-additive}, and \ref{quest:pseud-disk-exact-ply} affirmatively. 
We can even set the additive approximation constant $\beta$ as small as $1$, which is as close to the true diameter as one can get.  First, we present our results for unit-disk graphs.

\begin{theorem}\label{thm:main} There is an algorithm computing a $+1$-approximation of the diameter of any given unweighted unit-disk graph with $n$ vertices in $\Tilde{O}(n^{2-1/18})$ time. 
\end{theorem}

Our algorithm is a combination of two technical ingredients. 
(1) We show that both the distance encoding vectors defined by Le and Wulff-Nilsen~\cite{lw2024} as well as the set of $k$-neighborhood balls defined by Ducoffe, Habib, and Viennot~\cite{ducoffe2022diameter} have VC-dimension of $4$ for unit-disk graphs and pseudo-disk graphs in general.
(2) We develop a new \emph{clique-based $r$-clustering} which is analogous to an $r$-division for planar and minor-free graphs~\cite{Frederickson1987,Wulff-Nilsen11}. 
The combination is inspired by recent developments in computing diameter in truly-subquadratic time for minor-free graphs~\cite{lw2024}; we will discuss these technical ideas in detail in \Cref{subsec:ideas}. 
We then generalize our algorithm for unit-disk graphs to work with similar-size pseudo-disks with constant complexity. 

\begin{theorem}\label{thm:pseudodisk-diam-addive} 
Given an unweighted $n$-vertex similar-size pseudo-disk graphs with constant complexity, we can compute a $+1$-approximation of the diameter in $\Tilde{O}(n^{2-1/18})$ time.
\end{theorem}

In this general case, we need an additional component: a single-source shortest path (SSSP) algorithm with $\Tilde{O}(n)$ running time for the intersection graphs of similar-size pseudo-disks with constant complexity.  
SSSP algorithms with running time $\Tilde{O}(n)$ are known for some special cases, including unit-disk graphs~\cite{Efrat2001-hm,Chan2019-vw}, unit $L_p$-disks for $p = 1$ or $p = \infty$~\cite{Klost2023-xs}, and arbitrary disks~\cite{KKSS23}.

\medskip
When the objects have bounded ply (or even $n^{\delta}$-ply for small $\delta$), the intersection graphs have truly-sublinear separators, using the observation by 
de Berg \etal~\cite{BBKMv20} that the intersection graph of fat objects has sublinear clique-based separators. 
(Indeed, the objects in each clique of the clique-based separators are stabbed by a single point.) We use this fact combined with our VC-dimension result for pseudo-disk graphs to prove the following theorem.

\begin{theorem}\label{thm:pseudodisk-diam-exact}  Let $G$ be an unweighted $n$-vertex similar-size pseudo-disk graphs of with constant complexity, and let $k$ be the ply of $G$.  We can compute the exact diameter in $\Tilde{O}(k^{11/9}n^{2-1/18})$ time. 
\end{theorem}
The running time of \Cref{thm:pseudodisk-diam-exact} is truly subquadratic when $k = O(n^{1/22-\eps})$ for any constant $\eps$, including the special case of $k = O(1)$ as asked in \Cref{quest:pseud-disk-exact-ply}.

\medskip
Next, we showcase another application of our technique in constructing a \emph{distance oracle} for (unweighted) unit-disk graphs. 
The same technique in Chan and Skrepetos~\cite{Chan2019-vw} for the diameter problem mentioned above gives a distance oracle returning a hybrid $(1+\eps, +(4+2\eps))$-approximation of the shortest distance using $O(n\log^3 n)$ space and $O(1)$ query time. 
In the weighted setting, they got a multiplicative $(1+\eps)$-approximation with the same space and query time, improving upon an earlier result~\cite{Gao2005-kx}. 
Mark de Berg~\cite{deBerg23} considered the \emph{transmission graph} where each point has a transmission radius and can reach any vertex within the transmission radius. 
On this graph (which by definition is directed and unweighted), de Berg presented a distance oracle of size $\tilde{O}(n^{3/2}/\eps)$ that can answer approximate distance queries with a hybrid $(1+\eps,+1)$-approximation in time $\tilde{O}(n^{1/2}/\eps)$. 
The question is: can we develop a distance oracle with truly-subquadratic space and \emph{constant query time}, returning a purely additive approximation of shortest distances? 
We use the same technique developed for the diameter problem to answer this question positively.

\begin{theorem}\label{thm:main-oracle} Given an unweighted unit-disk graph with $n$ vertices, we can construct a distance oracle with $O(n^{2-1/18})$ space and $O(1)$ query time, returning a $+1$-approximation of the true distances. 
\end{theorem}

\Cref{thm:main-oracle} extends to pseudo-disk graphs as well. 

\begin{theorem}\label{thm:pesodo-dks-oracle}  
Given an unweighted $n$-vertex similar-size pseudo-disk graph with constant complexity, we can construct a distance oracle with $O(n^{2-1/18})$ space and $O(1)$ query time, returning a $+1$-approximation of the true distances. 
\end{theorem}

\subsection{Technical Ideas}\label{subsec:ideas}

Our technique is inspired directly by recent developments in computing exact diameters for minor-free graphs~\cite{ducoffe2022diameter,lw2024} that combine two well-known tools in geometric algorithms: \emph{VC-dimension} and \emph{$r$-division}.   An $r$-division is a decomposition of the graph into $\Theta(n/r)$ pieces, each with $O(r)$ vertices and $O(\sqrt{r})$ boundary vertices that are incident to other pieces. The result by Chepoi, Estellon, and Vax\`{e}s~\cite{Chepoi2007-ar} showed that the set of all $k$-neighborhood balls in a $K_h$-minor-free graph, when treated as a set system over the vertices, has VC-dimension at most $h-1$. Ducoffe \etal~\cite{ducoffe2022diameter} was the first to combine the VC-dimension result~\cite{Chepoi2007-ar} and $r$-division to design truly-subquadratic time algorithm for minor-free graphs. 
Le and Wulff-Nilsen~\cite{lw2024} designed 
a different VC set system based on that of Li and Parter~\cite{Li2019-li}, which is easier to combine with $r$-division. 
They obtained, among other things, an improved algorithm for computing exact diameter in minor-free graphs. 

\bigskip
\noindent We follow a path similar to the one taken for minor-free graphs~\cite{lw2024} to design an algorithm for unit-disk graphs.  To carry out this plan, we have to develop the two corresponding technical components in the geometric setting: An appropriate VC set system and an $r$-division for unit-disk graphs. 
There are two main challenges.
The first challenge is that
while the definitions of the VC set systems proposed in~\cite{Li2019-li,lw2024} are naturally applicable to any graphs, their proof technique heavily depends on graphs being minor-free (by building a minor directly from a system of high VC-dimension), and in some case involves tedious case analysis.
Our proof for unit disks only relies on their topological property of being \emph{pseudo-disks}.
The second challenge is rooted from the reality that $r$-division does not exist for unit-disk graphs. 
Here we introduced a new notion called \emph{clique-based $r$-clustering}, which allows cliques to be on the boundary of each region (called \emph{cluster} in our terminology). 
Our notion of clique-based $r$-clustering is inspired by clique-based balanced separators for geometric intersection graphs~\cite{BBKMv20,BKMT23}. 
However, formulating the right definition for clique-based $r$-clustering ends up to be delicate and challenging; the paragraphs after Remark~\ref{rm:DKP23} explain why several na\"ive approaches do not work, and our eventual solution.
We now elaborate on the two main technical components in more details.

\paragraph{Various definitions of VC-dimensions on graphs.}
In computational geometry literature, VC-dimension has been used to characterize the complexity and ``richness'' of geometric shapes~\cite{cw-qorss-1989}. 
The VC-dimension for unit-disks (as well as disks of all possible radii) is $3$ --- no four points can be shattered by disks in the plane~\cite{Matousek1990-gd}.
We remark that in prior work the \emph{VC-dimension of a graph} is defined on the set system of the closed \emph{immediate} neighborhoods, i.e., for each vertex $v$, the set of vertices including $v$ and its one-hop neighbors~\cite{Haussler1986-rr,Alon2006-wd,Bousquet15identifying}. In this definition, the VC-dimension of a unit-disk graph is $3$ as well~\cite{Bousquet15identifying}. 

Here we study the VC-dimension of two set systems: 
(1) the set of balls in the geometric intersection graph with radius $r$ ranges over all possible non-negative integers---this is referred to as the VC-dimension of the \emph{ball hypergraph} of $G$, also called the \emph{distance VC-dimension} of $G$~\cite{Chepoi2007-ar,Bousquet2015-pv,ducoffe2022diameter};
(2) the distance encoding vectors as defined in~\cite{lw2024} in a unit-disk graph with respect to a set $S$ of $k$ vertices. 
For both cases we show that the VC-dimension is exactly $4$ (not $3$) --- and we have an example of $4$ points that are shattered. 
In fact, we present a proof that is purely topological and thus can be generalized to the intersection graphs of pseudo-disks --- topological disks in the plane bounded by Jordan curves such that the boundaries of any two objects have at most two intersection points.  The pseudo-disk requirement is actually crucial and cannot be dropped. 
For example, we can construct $n$ unit-size equilateral triangles (possibly with rotations) with VC-dimension $\Omega(\log n)$ by modifying the fine-grained hardness construction in Bringmann \etal~\cite[Theorem 17]{Bringmann2022-me}.

\begin{remark}
\label{rm:DKP23} 
In~\cite{Abu-Affash2021-vs}
it is shown that unit-disk graph has distance VC-dimension $4$. As we completed our first technical component --- the VC-dimension results for pseudo-disk graphs --- we discovered an independent work posted on arXiv by Duraj, Konieczny and Pot\c{e}pa~\cite{duraj2023better}.  They showed that  geometric intersection graphs of objects that are closed, bounded, convex, and center symmetric has distance VC-dimension at most $4$, which is a subset of our result.  
Both of their proof techniques rely on geometry.  Our proof approach is purely topological, thus applies for general pseudo-disk graphs and distance encoding vectors. 

Duraj, Konieczny and Pot\c{e}pa~\cite{duraj2023better} combined their distance VC dimension bound with an (improved) argument along the lines of  Ducoffe, Habib, and Viennot~\cite{ducoffe2022diameter} to design truly-subquadratic time algorithms for intersection graphs of unit squares and translations of convex polygons with center of symmetry when the diameter is small. 
However, as noted above, it remains an open problem if the same result could be obtained for unit-disk graphs even when the diameter is small --- one missing element is a data structure that can efficiently build the $r$-neighborhood with increasing $r$. 
It is unclear if such a data structure could be constructed for unit-disk graphs. 
\end{remark}

\paragraph{Clique-based \boldmath{$r$}-clustering.}
As we mentioned above, $r$-division does not exist in unit-disk graphs. Here we develop an analogous clique-based $r$-clustering.  
A \EMPH{$\delta$-balanced clique-based separator} of a geometric intersection graph $G$~\cite{BBKMv20,BKMT23} is a collection $\mathcal{C}$ of vertex-disjoint cliques whose removal will partition the graph into two parts of size at most $\delta n$, with no edges between the parts. The \EMPH{clique size} of $\mathcal{C}$ is the number of cliques in $\mathcal{C}$, and the \EMPH{vertex size} of $\mathcal{C}$ is the total number of vertices in all cliques in $\mathcal{C}$.

As alluded to earlier, the definition and construction of the clique-based analog of $r$-division requires handling several subtleties. 
To explain these subtleties, we will suggest some natural ideas and discuss why these ideas do not work. 
\begin{itemize}
\item \emph{First attempt:}
Let $D$ be the input set of $n$ disks, whose intersection graph is $G$.  We could apply the clique-based  separator~\cite{BBKMv20,deBerg23} to find a set of $\sqrt{n}$ cliques $\mathcal{S}$ such that $D\setminus \mathcal{S}$ could be partitioned into sets $\{D_1,D_2,\ldots\}$ such that each set of disk $D_i$ has size at most $2n/3$ and induces a maximally connected intersection graph.  We call cliques in $\mathcal{S}$ \EMPH{boundary cliques} and disks in $\mathcal{S}$ \EMPH{boundary disks}.  Ideally, we want to recursively apply the clique-based separators to each set $D_i$ until we obtain the set of clusters $\mathcal{R}$ of size at most $r$ each. The issue here is that the number of boundary cliques adjacent to each region in $\mathcal{R}$ could be arbitrarily large, up to $\Omega(\sqrt{n})$.  Note that we want each set to have only $O(r)$ boundary cliques in the same way that $r$-division guarantees each region to have $O(r)$ boundary vertices. 

\item \emph{Second attempt:}
Instead of separating each $D_i$ directly, we could add the boundary disks in $\mathcal{S}$ back to $D_i$, and then recursively apply the clique-based separator theorem on each resulting $D_i$, as done in algorithms for constructing an $r$-division of planar and minor-free graphs~\cite{Frederickson1987,Wulff-Nilsen11}. There are several issues, and one of them is running time. Specifically, $\mathcal{S}$ could contain up to $\Omega(n)$ disks, and by reinserting the boundary disks across different $D_i$, the number of disks (counted with multiplicity) might be more than $n$, and hence the total number of disks arising over the course of the entire recursion could be up to $\Omega(n^2)$. 

\item \emph{Third attempt:}
One way to avoid adding too many boundary disks to $D_i$ is to add only \emph{one} boundary disk per clique in $\mathcal{S}$. 
Specifically, for each clique in $\mathcal{S}$, we choose a disk in the clique to be its \emph{representative}. 
Next, we add the representative of each clique to $D_i$, if the clique intersects at least one disk in $D_i$.  We then recursively apply the clique-based separator to the resulting set of disks. 
Here the total number of disks, counted with multiplicity, is $n + O(\sqrt{n})$ at the second level, and $O(n)$ over all levels. 

However, there is another technical issue with using representative disks of cliques in $\mathcal{S}$.  
Suppose that we apply the clique-based separator to $D_i$ (after adding the representative boundary disks) to find a clique-based separator $\mathcal{S}_i$. Removing $\mathcal{S}_i$  partitions $D_i$ into two balanced sets of disks  $X_1$ and $X_2$. 
There could be a representative disk $x \in D_i$ of a clique in $\mathcal{S}$ that is assigned to $X_1$ and not to $X_2$.  Yet, the clique represented by $x$ might contain a disk (other than $x$) that intersects disks in $X_2$. 
As $x$ is not in $X_2$, the algorithm does not correctly capture the boundary disks of $X_2$, and hence, when the algorithm terminates, the number of boundary cliques of each region could still be $\Omega(\sqrt{n})$.  
\end{itemize}

We ended up with the following (rather delicate) definition of a clique-based $r$-clustering.

\begin{definition}[Clique-based $r$-clustering]
\label{def:clique-r-division-2nd} Let $r\geq 1$ be a parameter. 
A \EMPH{clique-based $r$-clustering} of a geometric intersection graph $G$ is a pair $(\mathcal{R}, \mathcal{C})$ where $\mathcal{R}$ contains subsets of $V(G)$ called \EMPH{clusters}, 
and $\mathcal{C}$ is a set of vertex-disjoint cliques of $G$ such that: 
\begin{enumerate}
    \item Every set $R\in \mathcal{R}$ induces a connected subgraph of $G$.  Furthermore, $|\mathcal{R}| = O(n/\sqrt{r})$. 
    
    \item Every cluster $R\in \mathcal{R}$ can be partitioned into two parts, \EMPH{boundary $\bdry R$} and \EMPH{interior $R^{\circ}$}, such that all vertices in $R$ having neighbors outside $R$ belongs to $\bdry R$, and furthermore, (i) $R^{\circ}$ has at most $r$ vertices and  (ii) $\bdry R$ contains at most $r$ cliques in $\mathcal{C}$, denoted by \EMPH{$\mathcal{C}(\bdry R)$}.  
    
    \item  $\sum_{R\in \mathcal{R}}|\mathcal{C}(\bdry R)| = O(n/\sqrt{r})$. This in particular implies that $|\mathcal{C}| = O(n/\sqrt{r})$.
    \item  Every vertex of $G$ either belongs to a clique in $\mathcal{C}$ or  to $R^{\circ}$ for some cluster  $R\in \mathcal{R}$.
\end{enumerate}
\end{definition}

There are several differences between our clique-based $r$-clustering and an $r$-division in planar graph literature~\cite{Frederickson1987}. First of all, in our clique-based $r$-clustering we can no longer guarantee that each cluster $R\in \mathcal{R}$ has size at most $r$; we can only guarantee that its internal part $R^\circ$ has size at most $r$. 
Indeed, the size of $R$ could be $\Omega(n)$, thus computing an explicit representation of $\mathcal{R}$ could take $\Omega(n^2/\sqrt{r})$ time; thus, we only compute an implicit representation of $\mathcal{R}$. 
Second, the fact that $R$ could have size $\Omega(n)$ makes other algorithms relying on clique-based $r$-clustering more challenging:  we cannot go through every vertex of $R$ to do the computation in the way other planar algorithms do.  
Third, the number of cliques in the boundary of $R$ is $O(r)$ in the clique-based $r$-clustering, instead of $O(\sqrt{r})$ in a standard planar $r$-division.
Last but not least, we cannot simply compute a clique-based $r$-clustering from a balanced clique-based separator. 
Instead, we have to rely on a different kind of separator, called a \emph{well-separated clique-based separator}.  
The basic idea is that we can find a balanced clique-based separator such that the remaining disks can be partitioned into two sets that are far from each other relative to the radii of the disks. We defer the details to \Cref{sec:clique-division}.
We show that well-separated clique-based separators exist for unit-disk graphs or fat pseudo-disks of roughly the same size.

Now, we state our algorithm for computing a clique-based $r$-clustering. We will compute an \EMPH{implicit representation} of $\mathcal{R}$: 
for each clique in $\mathcal{C}$, we will choose an arbitrary vertex to be the \EMPH{representative} of the clique, and for each cluster $R\in \mathcal{R}$, we explicitly store vertices in $R^{\circ}$ and all representatives of the cliques in $\mathcal{C}(\bdry R)$, denoted by \EMPH{$\rep(R)$}. 
Furthermore, for each vertex $u\in R^{\circ}$, we will maintain a list of representatives $x$ by which $u$ has a neighbor in the clique represented by $x$. 

 \begin{lemma}
 \label{lm:clique-division} 
     For any given integer $r$ and an $n$-vertex unit-disk graph $G$, we can find the implicit representation of a clique-based $r$-clustering $(\mathcal{R},\mathcal{C})$ of $G$ in $O(n \log^2 n)$ time.
 \end{lemma}

\paragraph{Pseudo-disk graphs with constant ply.}  
Ducoffe \etal~\cite{ducoffe2022diameter} showed that if a monotone class of graphs $\mathcal{G}$ has truly-sublinear balanced separators and distance VC-dimension at most $d$, then we can compute the diameter in time $O(n^{2-\eps_{\mathcal{G}}(d)})$ where $\eps_{\mathcal{G}}(d) =  1/2^{O_{\mathcal{G}}(d)}$; the $O_{\mathcal{G}}(\cdot)$ notation hides a dependency on the family $\mathcal{G}$.  
Our results above imply that the family of intersection graphs of similar-size pseudo-disks of constant complexity and ply has truly-sublinear balanced separators and distance VC-dimension at most $d$. 
Thus, we can solve diameter exactly in time $O(n^{2-\eps_{k}(d)})$ where the constant $\eps_k(d)$ depends on the ply $k$ using the algorithm of Ducoffe \etal~\cite{ducoffe2022diameter} as a black box. 
However in their algorithm the dependency on $k$ is not explicitly computed, and furthermore, the dependency on $d$ is exponentially diminishing. 
Instead, we modify our approximation algorithm for unit-disk graphs to obtain a better dependency on $k$ and a smaller constant in the exponent of $n$.  The basic idea is that we could now use $r$-division instead of clique-based $r$-division. We note that $r$-division for low-ply geometric intersection graphs was used earlier to solve different problems by Har-Peled and Quandrud~\cite{HQ17}. 

\paragraph{From \boldmath{$+2$}-approximation to \boldmath{$+1$}-approximation.} 
In the conference version~\cite{Chang2024-ab}, we had a $+2$-approximation for both graph diameter and distance oracles for unit-disk graphs and pseudo-disk graphs. In this version, we have improved the additive constant to $+1$. The main difference is to choose a dummy vertex to represent each clique with edges of weight $1/2$ to other vertices in a clique. This allows for the saving of additive error of $1$ but requires some minor (but unintuitive) edits of the overall arguments. To keep the argument clean, we present all results with $+2$ approximation and in \Cref{sec:plusone} we present the improvement to $+1$. We believe that presenting the +1-approximation of diameter in its full generality, without the +2-appproximation algorithm as an important context, will seriously hinder the readability of this paper.

\subsection{Additional Related Work}\label{sec:related}

\paragraph{{Diameter} in General Graphs.} 
Finding the diameter of a given graph can be easily done by computing all-pairs shortest paths (APSP) in $O(n^3)$ time using the classical Floyd-Warshall algorithm, or in time $O(n^3/2^{\Omega(\log n)^{1/2}})$~\cite{Williams2014-mj} after a long line of improvement of polylogarithmic factor; see \cite{Williams2014-mj} for a historical discussion.  
No truly-subcubic time algorithm is known for either all-pairs shortest paths or for computing the graph diameter. It is also not clear if computing diameter is as hard as APSP\/.  If the edges are unweighted, computing the diameter can be done in time $\tilde{O}(n^{\omega})$ where $\omega$\footnote{The recent bounds are: $\omega<2.37286$~\cite{Alman2021-np}, $\omega<2.371866$~\cite{duan2023faster} and $\omega<2.371552$~\cite{williams2023new}.} is the exponent of the running time of matrix multiplication~\cite{Seidel1995-hn}.

In the sparse setting, when the graph has only a linear number of edges, one can run single-source shortest paths algorithm (SSSP) from each vertex, achieving $O(n^2)$ running time, or even $O(n^2/\log n)$~\cite{Chan2012-gn} by compressing the bits; none of these algorithms are truly subquadratic. 
In fact, assuming strong exponential time hypothesis (SETH)~\cite{Impagliazzo2001-sr}, there is no truly-subquadratic algorithm for computing the diameter of a graph using a reduction from the orthogonal vector problem --- even distinguishing between $2$ and $3$~\cite{Roditty2013-zr}.  For many special graphs including planar graphs and graphs with forbidden minors,
one can find subquadratic algorithms.
We will review these results below.

\paragraph{Planar Graphs.} 
\textsc{Diameter} is first shown, in a breakthrough paper by Cabello~\cite{Cabello2018-gz}, to be solvable for planar graphs within time $\tilde{O}(n^{11/6})$ and later improved to $\tilde{O}(n^{5/3})$ \cite{Gawrychowski2018-zy}. 
Both algorithms use two major elements:  (1) the $r$-division by Frederickson~\cite{Frederickson1987}
and (2) for each vertex $v_0$ and each piece $P$ build an additive Voronoi diagram within $P$ with boundary vertices as sites and each Voronoi cell containing vertices that share the same boundary vertex on their shortest paths to $v_0$.  While the $r$-division can be efficiently computed in time $O(n)$ for planar graphs~\cite{Klein2013-xv}, computing the additive Voronoi diagrams efficiently requires a lot of technicalities.

Li and Parter~\cite{Li2019-li} addressed distributed algorithms for \textsc{Diameter} in planar graphs and avoided using the abstract Voronoi diagrams. 
Instead, they used the approach of \emph{metric compression} --- intuitively, given a sequence of $k$ vertices $S=\Seq{s_1,\dots,s_k}$, for each vertex $v$ define a set of tuples $\{(i, \Delta)\}$ with $\Delta$ being an upper bound on the difference of distances $d(v, s_i)$ and $d(v, s_{i-1})$. These distance vectors encode (approximately) the distance from $v$ to each vertex in $S$.  For \textsc{Diameter}, Li and Parter use $S$ as vertices on a cycle separator of the planar graph. Thus, the distance encoding vectors provide a compression of all shortest path distances from $V$ with the separator $S$. Due to planarity, this set system of distance encoding vectors has VC-dimension at most $3$. Therefore the size of distinct tuples is polynomially bounded in the size of $S$, which is crucial for bounding computation time.

\paragraph{Approximating \textsc{Diameter}.} 
Approximating the diameter in weighted (di)graphs can be done in $\tilde{O}(m)$ time for a $2$-approximation~\cite{HENZINGER19973} or in time $\tilde{O}(m^{3/2})$ for a $(3/2)$-approximation~\cite{Berman2007-rx,Roditty2013-zr, Aingworth1999-qc, Chechik2013-gh}. Assuming SETH, on an undirected unweighted graph any $(3/2-\eps)$-approximation with $\eps>0$ of the diameter requires time $n^{2-o(1)}$~\cite{Roditty2013-zr}, any $(5/3-\eps)$-approximation requires time $n^{3/2-o(1)}$~\cite{Li21settling}, and any $(7/4-\eps)$-approximation requires time $n^{4/3-o(1)}$~\cite{Bonnet2022-xb}.
For a weighted undirected \emph{planar} graph with non-negative edge weights, $(1+\eps)$-approximation to the diameter can be done in running time near-linear in $n$ (but exponential in $1/\eps$)~\cite{Weimann2015-py} and later improved to $O(n \log n \cdot (\log n + (1/\eps)^5))$~
\cite{Chan2019-rh} and to $O_\e(n \log n)$ time~\cite{CKT22}.

\section{VC-dimension of Unit-Disk and Pseudo-Disk Graphs}

\subsection{Unit-Disk Graphs and Pseudo-Disk Graphs}

An undirected, unweighted \emph{unit-disk graph} is a graph obtained from a set of points $P$ in the plane such that two points are connected by an edge if and only if their Euclidean distance is at most $2$. 
A unit-disk graph is a special type of \emph{geometric intersection graph}, which can be defined for a set $F$ of objects in $\reals^d$ where an edge exists between two vertices if and only if the two corresponding objects overlap. 

One interesting family of geometric intersection graphs is when the objects are \emph{pseudo-disks}. 
Specifically, a simple closed Jordan curve $C$ partitions the plane into two regions, one of them is bounded, called the \EMPH{interior} of $C$. 
A family of simple closed Jordan curves is called \EMPH{pseudo-circles} if every two curves are either disjoint or properly crossed at precisely two points. 
(Without loss of generality we assume there are no tangencies.)
In a family of pseudo-circles, the interior of each pseudo-circle is called a \EMPH{pseudo-disk}. 
Each pseudo-disk is a simply connected set and the intersection of a pair of pseudo-disks is either empty or is a connected set~\cite{Buzaglo2013-lj}.
For a family of pseudo-disks $D_1, D_2, \ldots, D_n$, we can construct the intersection graph $G$ of the pseudo-disks --- combinatorially, we use a set of vertices with $v_i$ corresponding to $D_i$ and connect an edge for $v_i$ and $v_j$ if and only if $D_i$ and $D_j$ have non-empty intersections. 

The following property of unit-disk graphs is folklore (for example, see Breu~{\cite[Lemma~3.3]{Breu-thesis}}).

\begin{lemma}
\label{lem:intersect}
If two edges $ab$ and $cd$ in a unit-disk graph intersect, then one of the four vertices $a, b, c, d$ is connected to the rest of three vertices.
\end{lemma}

We now prove an analog of Lemma~\ref{lem:intersect} for pseudo-disks. The proof uses only the topological properties of pseudo-disks. 

\begin{claim}\label{clm:pseudodisk-inter}If two pseudo-disks $D$ and $D'$ intersect, for any two points $p\in D$ and $p'\in D'$, we can find a curve $\pi(p, p')$ from $p$ to $p'$ inside $D \cup D'$ such that this path can be partitioned into three pieces, at point $q, q'$ with $\pi(p, q)\in D\setminus D'$, $\pi(q, q')\in D\cap D'$ and $\pi(q', p')\in D'\setminus D$. Any of the three pieces may be empty. 
\end{claim}
\begin{proof}
     We first take any curve $\pi(p, p')$ inside $D \cup D'$ and take $q$ to be the first point on the curve that enters $D'$ and $q'$ the last point on the curve that leaves $D$. This partitions the curve into three pieces,  $\pi(p, q)$, $\pi(q, q')$ and $\pi(q', p')$. Clearly $\pi(p, q)\in D\setminus D'$ and $\pi(q', p')\in D'\setminus D$ by definition. Now if $\pi(q, q')$ is not entirely inside $D\cap D'$, we replace it by a curve purely inside $D\cap D'$ since $D\cap D'$ is connected. See Figure~\ref{fig:pseudo-disk-path} for an example.
\end{proof}
We call the curve $\pi(p, p')$ in Claim~\ref{clm:pseudodisk-inter} a \EMPH{proper curve}  connecting $p$ and $p'$, and $q, q'$ the \EMPH{entrance} and \EMPH{exit} point respectively. 

\begin{figure}[htbp]
\centering
\includegraphics[width=0.375\linewidth]{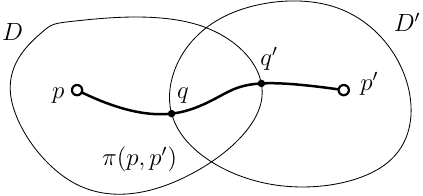}
\internallinenumbers
\caption{If two pseudo-disks $D$ and $D'$ intersect, for any two points $p\in D$ and $p'\in D'$, we can find a curve $\pi(p, p')$ from $p$ to $p'$ inside $D \cup D'$ such that this path can be partitioned into three pieces, at point $q, q'$ with $\pi(p, q)\in D\setminus D'$, $\pi(q, q')\in D\cap D'$ and $\pi(q', p')\in D'\setminus D$. }
\label{fig:pseudo-disk-path}
\end{figure}

\begin{lemma}[Lemma 1 in~\cite{Buzaglo2013-lj}] \label{lem:pseudo-even}
Let $\gamma$ and $\gamma'$ be arbitrary non-overlapping curves contained in pseudo-disks $D$ and $D'$\!, respectively. If the endpoints of $\gamma$ lie outside of $D'$ and the endpoints of $\gamma'$ lie outside of $D$, then $\gamma$ and $\gamma'$ cross an even number of times.
\end{lemma}

\begin{lemma}\label{lem:intersect-pseudo} 
For four pseudo-disks $\{D_a, D_b, D_c, D_d\}$ with $D_a$ intersects $D_b$ and $D_c$ intersects $D_d$, take four points $a \in D_a$, $b\in D_b$, $c\in D_c$ and $d\in D_d$ and proper curves $\pi(a, b), \pi(c, d)$. If  $\pi(a, b), \pi(c, d)$ have an odd number of intersections and none of point $i\in \{a, b, c, d\}$ stays inside any pseudo-disk $D_j$ with $j \in \{a, b, c, d\}$ and $j\neq i$, then one of the four pseudo-disks intersects all three other pseudo-disks.
\end{lemma}
\begin{proof}
Suppose there are $k$ intersections of $\pi(a, b), \pi(c, d)$, where $k\geq 1$ is an odd number.  First, suppose at least one of the intersections of $\pi(a, b), \pi(c, d)$, say $w$, stays in between the entrance and exit of $\pi(a, b)$. Recall that $w$ also stays on $\pi(c, d)$ and thus stays either inside $D_c$ or $D_d$. This means that either $w$ stays inside all three pseudo-disks $D_a, D_b, D_c$ (which means that $D_c$ intersects all three other pseudo-disks) or inside all three pseudo-disks  $D_a, D_b, D_d$  (which means that $D_d$ intersects all three other pseudo-disks). The same argument can be applied if one intersection $w$ stays in between the entrance and exit of $\pi(c, d)$. 

Now consider the $k$ intersections on $\pi(a, b)$, none of them stays in between the entrance $q_{ab}$ and exit $q'_{ab}$. Similarly, none of these $k$ intersections stays in between the entrance $q_{cd}$ and exit $q'_{cd}$. Define $k_{ac}$ to be the number of intersections of $\pi(a, q_{ab})$ and $\pi(c, q_{cd})$. Define $k_{bc}, k_{ad}, k_{bd}$ in a similar manner.
We have $k=k_{ac}+k_{bc}+k_{ad}+k_{bd}$. Since $k$ is odd, at least one of the four numbers is odd. Without loss of generality, assume that $k_{ac}$ is odd, i.e., $\pi(a, q_{ab})$ and  $\pi(c, q_{cd})$ intersect each other an odd number of times. 
Now $\pi(a, q_{ab})$ is entirely in $D_a$ and $a$ is outside $D_c$. If $q_{ab}$ is inside $D_c$ we have $D_c$ intersecting both $D_a, D_b, D_d$ and we are done. Thus we assume that $q_{ab}$ is also outside of $D_c$. By the same argument $\pi(c, q_{cd})$ is entirely inside $D_c$ and both end points $c$ and $q_{cd}$ are outside of $D_a$. $\pi(a, q_{ab})$ and $\pi(c, q_{cd})$ intersect each other an odd number of times. This contradicts Lemma~\ref{lem:pseudo-even} and therefore is not possible. 
\end{proof}

Now we consider a pseudo-disk graph. 
We can find a planar drawing of the pseudo-disk graph in the plane in the following manner: 
for each pseudo-disk $D$, we take one \EMPH{representative point} $p\in D$. 
If two pseudo-disks $D$ and $D'$ intersect, we connect their representative points $p\in D$ and $p'\in D'$ by a proper curve $\pi(p, p')$. 
This graph is unweighted, i.e., the proper curve $\pi(p, p')$ has length of $1$.  
\EMPH{Path $P(p, q)$} for two pseudo-disks represented by $p$ and $q$ consists of several proper curves visiting the representative points of the pseudo-disks on the path. We use $|P(p, q)|$ to denote the hop length of a path $P(p, q)$. 

We can prove a generalized version of Lemma~\ref{lem:intersect-pseudo} for the hop distances of paths in the pseudo-disk graph. This will be useful for bounding the VC-dimension of set systems defined on the pseudo-disk graph. 
Consider four vertices $a, b, c, d$ representing four pseudo-disks $D_a, D_b, D_c, D_d$ and assume that there are two paths $P(a, b)$ and $P(c, d)$. We define a \EMPH{local crossing pattern} to be four distinct vertices \EMPH{$a', b', c', d'$} with $a', b'$ on path $P(a, b)$ (with $a'$ closer to $a$ than $b'$) and $c', d'$ on path $P(c, d)$ (with $c'$ closer to $c$ than $d'$) such that one of the four vertices $a', b', c', d'$ has edges to all the other three vertices. 

\begin{lemma}\label{lem:intersect-path}
Consider four vertices $a, b, c, d$ representing four pseudo-disks $D_a, D_b, D_c, D_d$ and assume that there is a local crossing pattern of the two paths $P(a, b)$ and $P(c, d)$, then the followings are true:
\begin{enumerate}
    \item Either there is a path $P'(a, c)$ whose hop length is at most $|P(c, d)|$ or there is a path $P'(b, d)$ whose hop length is at most $|P(a, b)|$.
    \item Either there is a path $P'(a, c)$ whose hop length is at most $|P(a, b)|$ or there is a path $P'(b, d)$ whose hop length is at most $|P(c, d)|$.
\end{enumerate} 
\end{lemma}
\begin{proof}
If one vertex in the local crossing pattern stays on both paths $P(a, b)$ and $P(c, d)$, denote this vertex as vertex $o$.
Then we take the path $P'(a, c)$ that is composed of $P(a, o)$ (along $P(a, b)$) and $P(o, c)$ (along $P(d, c)$), and the path $P'(b, d)$ that is composed of $P(b, o)$ (along $P(b,a)$) and $P(o, d)$ (along $P(c, d)$). By definition $|P'(a, c)|=|P(a, o)|+|P(o,c)|$, $|P'(b, d)|=|P(b, o)|+|P(o, d)|$. If $|P(o,a)| \leq |P(o, d)|$, then $|P'(a, c)|\leq |P(c, d)|$. Otherwise, $|P'(b, d)| \leq |P(a, b)|$. Similarly, if $|P(o,c)| \leq |P(o, b)|$, then $|P'(a, c)|\leq |P(a, b)|$. Otherwise, $|P'(b, d)| \leq |P(c, d)|$.

From now on we assume that the two paths $P(a, b)$ and $P(c, d)$ do not share any pseudo-disks. Without loss of generality, assume that $a'$ has edges to $b', c' d'$. Take the subpath from $a$ to $a'$ on the path $P(a, b)$ to be $P(a, a')$. We assume that $|P(a, a')|=d_1$, $|P(b',b)|=d_2$, $|P(c, c')|=d_3$ and $|P(d, d')|=d_4$. 
See Figure~\ref{fig:intersection-path} for an example.
Note that any of $d_1, d_2, d_3, d_4$ could be zero. Therefore, 
\begin{alignat*}{3}
|P(a, b)| &={} &|P(a, a')|+|P(a', b')|+|P(b', b)| &\geq{} &d_1+d_2+1, \\
|P(c, d)| &={} &|P(c, c')|+|P(c', d')|+|P(d', d)| &\geq{} &d_3+d_4+1.  
\end{alignat*}

\begin{figure}
\centering
\includegraphics[width=0.34\linewidth]{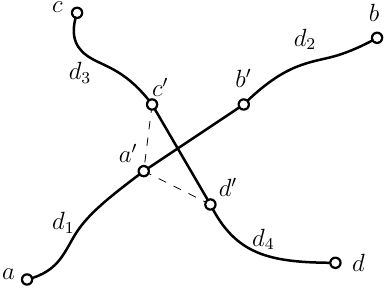}
\caption{If two paths $P(a, b)$ and $P(c, d)$ intersect with a local crossing pattern $a', b', c', d'$, then there is a path between $a, c$ that are no longer than $P(a, b)$ or there is a path between $b, d$ that is no longer than $P(c, d)$. }
\label{fig:intersection-path}
\end{figure}

Consider the path $P(a, a')$ followed by an edge $a'c'$ and then path $P(c',c)$. This is a path $P'(a, c)$ that has length $d_1+1+d_3$. Similarly, there is a path that connects $d$ to $b$ by using $P(d,d')$, the edge $d'a'$, path $P(a', b')$, and then path $P(b', b)$. This path, called $P'(d, b)$ has length $d_4+1+|P(a', b')|+d_2$. We argue that either $|P'(a, c)|\leq |P(c, d)|$ or $|P'(d, b)|\leq |P(a, b)|$. If otherwise, we have 
\begin{alignat*}{4}
|P'(a, c)| &= d_1+1+d_3 &\,>{} &|P(c, d)| &\,\geq{} &d_3+d_4+1 &\quad\Rightarrow\quad &d_1 > d_4, \\
|P'(d, b)| &= d_4+d_2+1+|P(a', b')| &\,>{} &|P(a, b)| &\,={} &d_1+|P(a', b')|+d_2 &\quad\Rightarrow\quad &d_4 +1 > d_1.
\end{alignat*}
Recall that $d_1, d_4$ must be integers thus it is impossible to have $d_4+1> d_1>d_4$.

Similarly, we argue that either $|P'(a, c)|\leq |P(a, b)|$ or $|P'(d, b)|\leq |P(c, d)|$. If otherwise, we have 
\begin{alignat*}{4}
|P'(a, c)| &= d_1+1+d_3 &\,> &\,|P(a, b)| &\,= &\,d_1+|P(a', b')|+d_2 &\quad\Rightarrow\quad & d_3 +1> d_2+|P(a', b')|, \\
|P'(d, b)| &= d_4+d_2+1+|P(a', b')| &\,> &\,|P(c, d)| &\,\geq &\,d_3+d_4+1 &\quad\Rightarrow\quad & d_2 +|P(a', b')| > d_3.
\end{alignat*}
Again this is impossible for $d_2, d_3$ to have integer values to satisfy $d_3 +1> d_2+|P(a', b')|> d_3$.
\end{proof}

\begin{remark} 
Lemma~\ref{lem:intersect-path} does not require the paths to be shortest. Furthermore, 
the argument in  Lemma~\ref{lem:intersect-path} is not true for weighted unit-disk graphs, where edges are given natural weights as the Euclidean length of the edges.     
\end{remark}

\subsection{VC-dimension of Unit-Disk Graphs and Pseudo-Disk Graphs}

The VC-dimension of a set system $(P, \mathcal{R})$ with $\mathcal{R}$ containing subsets of $P$ is the largest cardinality of a subset $S\subseteq P$ that can be \emph{shattered}, i.e., all subsets of $S$ can be obtained by the intersection of some sets in $\mathcal{R}$ with $P$. Here we consider the VC-dimension of two other set systems defined on a unit-disk graph or a pseudo-disk graph $G$, namely the \emph{distance VC-dimension} of $G$ and the \emph{distance encoding VC-dimension} of~$G$. 
We discuss them separately. 

\paragraph{Distance VC-dimension.}
In an unweighted graph $G$, consider the collection of balls $B(v, r)$, which is the set of all points within a hop distance of $r$ from a vertex $v$. 
Since we consider unweighted graphs, we assume $r$ to be non-negative integers. We define the \EMPH{ball system} of a graph $G$ on points $P$ as the sets 
\[
\EMPH{$\mathcal{B}(G)$} \coloneqq \Set{\big. B(v, r) : \text{$\forall v \in P$, $\forall r\in \mathbb{Z}$, $r\geq 0$} }.
\]

The VC-dimension of the set of balls with radius $r$ for all possible non-negative integers is referred to as the VC-dimension of the ball hypergraph of $G$, also called the \EMPH{distance VC-dimension} of $G$~\cite{ducoffe2022diameter}. 
It is known that the set system of balls of any undirected (weighted) $K_h$-minor-free graphs have VC-dimension at most $h-1$~\cite{Chepoi2007-ar}. Thus the set of balls for planar graphs has VC-dimension at most $4$, since a planar graph does not have $K_5$ as a minor. 
Notice that a unit-disk graph can be a complete graph thus is not $K_h$-minor-free for any $h\leq n$. 
Thus the above result does not immediately apply to a unit-disk graph.
Ducoffe \etal~\cite{ducoffe2022diameter} also showed that interval graphs have distance VC-dimension of two. 
Unit-disk graph is a natural extension of the interval graph to two dimensional space. Our \Cref{thm:distance-VC} below shows that the distance VC-dimension of a pseudo-disk graph is $4$. 
Figure~\ref{fig:4pt-shatter} is an example of $4$ points in a unit-disk graph that can be shattered. 
In fact, the same example shows that the distance VC-dimension of planar graphs is exactly 4.

\begin{figure}[t]
\centering
\includegraphics[width=0.5\linewidth]{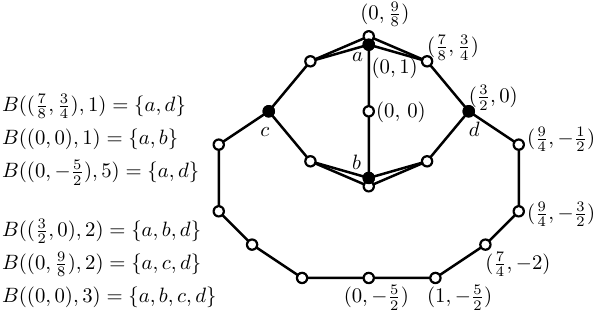}
\internallinenumbers
\caption{An example of $4$ points (drawn in solid) that can be shattered. The coordinates of the points are given and all edges of the unit-disk graph are drawn. Some examples of balls that shatter some subsets are given on the side. The remaining cases can be obtained by symmetry.}
\label{fig:4pt-shatter}
\end{figure}

\begin{theorem}\label{thm:distance-VC}
The distance VC-dimension of a pseudo-disk graph is $4$.
\end{theorem}
\begin{proof}
We just need to show that the largest set that can be shattered by the balls is $4$. 
Equivalently, we show that any set of $5$ vertices cannot be shattered. For five vertices $a, b, c, d, e$, we assume that they can be shattered; that is, for any subset $S\subseteq \{a, b, c, d, e\}$, there is a vertex $v_{S}$ with radius $r_{S}$ which includes all vertices of $S$ but not vertices in $\{a, b, c, d, e\}\setminus S$.  
We argue for a contradiction.

\begin{figure}[t]
\centering
\includegraphics[width=0.35\linewidth]{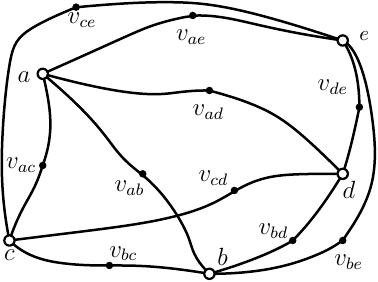}
\caption{$P(v_{ab},b)$ and path $P(v_{cd}, c)$ intersect.}
\label{fig:5pt}
\end{figure}

For each pair among $a, b, c, d, e$, we take the ball that separates this pair from the rest of the points. 
For example, the ball centered at $v_{ab}$ with radius $r_{ab}$ includes $a, b$ but not $c, d, e$. Thus we find a path $P(a, b)$ that is composed of the shortest path from $a$ to $v_{ab}$ and the shortest path from $v_{ab}$ to $b$. Similarly, we define paths for all  $5 \choose 2$  pairs of vertices. This becomes a graph $G$ that has $K_5$ as a minor. Thus $G$ is not planar. 
By Hanani–Tutte theorem~\cite{Hanani1934,Tutte1970,Schaefer2013-bf}, every drawing of $G$ in the plane contains a pair of paths $P(a, b)$ and $P(c, d)$, not sharing endpoints, that cross each other an odd number of times.

Let $G$ be a graph with vertices as representative points of the pseudo-disks involved and edges as proper curves connecting two neighboring pseudo-disks on the $5 \choose 2$ paths $P(x, y)$ with $x, y\in \{a, b, c, d, e\}$. 
We now re-arrange the representative points of the pseudo-disks and the proper curves (edges) of this graph $G$, but keep exactly the same planar drawing. Basically, for a path $P(x, y)$ with $x, y\in \{a, b, c, d, e\}$ and $x$ lexicographical earlier than $y$, keep the same drawing of the path $P(x, y)$ but we move the representative point of any pseudo-disk $z$ on $P(x, y)$, $z\neq x$ to be the exit point of the proper curve $\pi(z', z)$ with $z'$ as the preceding pseudo-disk of $z$ on $P(x, y)$. Essentially the representative point $z$ is just shifted forward along the path $P(x, y)$ to be on the boundary of the previous pseudo-disk --- like moving beads along a necklace. 
For a pseudo-disk $x\in \{b, c, d, e\}$, initially the paths $P(x, y)$ with different $y$ are joined at the representative point of $x$. Now they are still joined at the shifted representative point $x$, which is moved to the last exit point on path $P(a, x)$.
Last, we move the representative point $a$ to be the entrance point on the proper curve connecting $a$ with the next pseudo-disk on $P(a, b)$. 
After this re-arrangement, any representative point stays in at least two pseudo-disks. See Figure~\ref{fig:rearrange}.

\begin{figure}[t]
\centering
\includegraphics[width=0.92\linewidth]{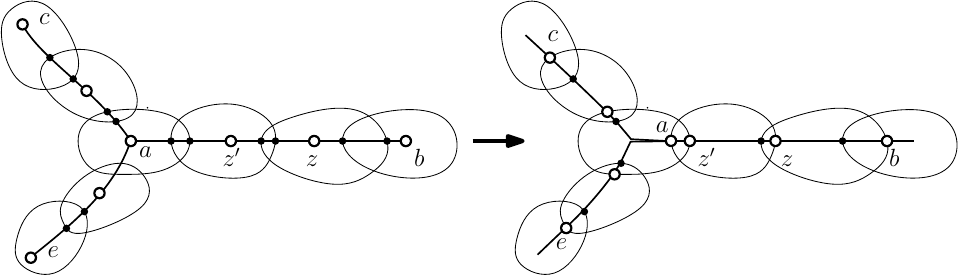}
\caption{Re-arrangement (Right) of the representative points (Left) in $G$.}
\label{fig:rearrange}
\end{figure}

Now for each path $P(x, y)$ of graph $G$, we partition it into pieces: the pieces that stay within at least two neighboring pseudo-disks on $P(x, y)$ are called \EMPH{multi-covered}, and the (open) pieces that are only inside one pseudo-disk are called \EMPH{single-covered}.   Every single-covered piece has two endpoints each staying in at least two pseudo-disks. And with the re-arrangement, each representative point is one of such endpoints. 
We emphasize that single-covered pieces may still intersect pseudo-disks outside~$P(x, y)$.

\medskip
To finish the proof, we prove two claims. First, in a planar drawing of $G$ there will be crossings of two paths that lead to a local crossing pattern, specifically, there are two neighboring pseudo-disks on each path and one pseudo-disk has edges to all the other three pseudo-disks. Second, this local crossing pattern leads to a contradiction.   

First we prove the second claim. 
Let's consider two paths $P(a,b)$ and $P(c, d)$. If there is a local crossing pattern anywhere on the two paths, we
have a contradiction. 
Indeed, say there are four distinct vertices $a', b', c', d'$ with $a', b'$ on $P(a,b)$ and $c', d'$ on $P(c, d)$ with one vertex having edges to all the other three vertices.
Without loss of generality, suppose $a', b'$ stay on path $P(v_{ab},b)$ and $c', d'$ stay on path $P(v_{cd}, c)$. 
By Lemma~\ref{lem:intersect-path}, either $d(v_{ab}, c)\leq d(v_{ab}, b)\leq r_{ab}$
or $d(v_{cd}, b)\leq d(v_{cd}, c)\leq r_{cd}$.
This means that either $c$ is in the ball $B(v_{ab}, r_{ab})$, or $b$ is in the ball $B(v_{cd}, r_{cd})$; either way, a contradiction.

Now we argue that in a planar drawing of $G$, there must be some local crossing patterns. 
First, any intersection $p$ of two path $P_1, P_2$ of $G$ (guaranteed by the Hanani-Tutte theorem) on a multi-covered piece will lead to a local crossing pattern --- if $p$ stays in two neighboring pseudo-disks $D_1$ and $D_2$ on path $P_1$ and one pseudo-disk $D'_1$ on path $P_2$, $D'_1$ has edges to both $D_1$ and $D_2$ from $P_1$ and an edge to one neighboring pseudo-disk on $P_2$. This is a local crossing pattern we are looking for.
Now we can assume that all crossings of paths in $G$ happen between pairs of single-covered pieces. 
Since the total number of crossings  between $P_1$ and $P_2$ is odd, there must be at least one pair of single-covered pieces $\pi_1$ and $\pi_2$ (not sharing endpoints) that intersect an odd number of times. 
Suppose $\pi_1$ is in pseudo-disk $D_1$ of path $P_1$ and $\pi_2$ is in pseudo-disk $D_2$ of path $P_2$. By Lemma~\ref{lem:pseudo-even}, if the endpoints of $\pi_1$ are outside of $D_2$ and the endpoints of $\pi_2$ are outside of $D_1$, we have a contradiction. Therefore at least one endpoint say $w$ of $\pi_1$ is inside $D_2$. 
Point $w$ is inside three pseudo-disks: $D_1$, $D'_1$ on path $P_1$, and $D_2$ on path $P_2$. See Figure~\ref{fig:crossing}.
Thus $D_2$ on path $P_2$ has edges to two neighboring pseudo-disks on $P_1$. $D_2$ also has a neighboring disk on $P_2$. This is a local crossing pattern. 

\begin{figure}
\centering
\includegraphics[width=0.35\linewidth]{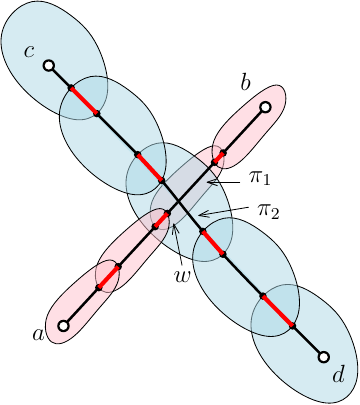}
\caption{We highlight in red the multi-covered segments on a path that stay within two or more pseudo-disks along the path. The intersection happens between $\pi_1$ on path $P(a, b)$ and $\pi_2$ on $P(c, d)$ with one endpoint $w$ of $\pi_1$ inside a pseudo-disk $D_2$ on $P(c, d)$. This triggers a local crossing pattern.}
\label{fig:crossing}
\end{figure}
Last, Figure~\ref{fig:4pt-shatter} is an example of $4$ points that can be shattered. 
Therefore the VC-dimension of a pseudo-disk graph is exactly $4$.
\end{proof}

\paragraph{Distance encoding VC-dimension.}
Li and Parter~\cite{Li2019-li} defined a distance encoding function in a graph and used it for computing diameter in a planar graph. 
Later Le and Wulff-Nilsen~\cite{lw2024} used a slighly revised one. We take the definition by  Le and Wulff-Nilsen~\cite{lw2024} and argue that this set defined on a unit-disk graph also has low VC-dimension.

\begin{definition}
Let $M\subseteq \reals$ be a set of real numbers. Let $S=\{s_0, s_1, \ldots, s_{k-1}\}$ be a sequence of $k$ vertices in an undirected weighted graph $G=(V, E)$. 
For every vertex $v$ define 
\[
\EMPH{$X(v)$} \coloneqq \Set{\big. (i, \Delta) : 1\leq i \leq k-1, \Delta \in M, d(v, s_i)-d(v, s_0)\leq \Delta }.
\]
Let $\EMPH{$\ecd$} \coloneqq \{X(v) : v\in V\}$ be a set of subsets of the ground set $[k-1]\times M$.
\end{definition}

The set $\ecd$ is a set of ``ranges'' where each range $X(v)$ corresponds to a vertex $v \in V$, which captures the distance to vertices in $S$ compared to the distance to $s_0$. We note that $s_{i}$ and $s_{i+1}$ may not be adjacent in $G$.

\begin{theorem}\label{thm:LW-dim}
Let $S$ be as any set of vertices of a pseudo-disk graph $G$ and $M \subseteq \reals$ be any set of real numbers. $\ecd$ has VC-dimension at most $4$.  
\end{theorem}
\begin{proof}
The proof is by contradiction. Suppose there is a set $Y$ of size $5$ that is shattered by $\ecd$. Without loss of generality let $Y=\{(s_1, \Delta_1), \cdots, (s_5, \Delta_5)\}$. 

By definition, that $Y$ is shattered, \emph{no two tuples share the same vertex}, i.e., $s_i \neq s_j$, for $1\leq i, j \leq 5$. If otherwise, suppose we have $s_1=s_2$. Without loss of generality, suppose $\Delta_1 < \Delta_2$. Since $Y$ is shattered, that is a set $X(v)$ such that $X(v) \cap Y=\{ (s_1, \Delta_1)\}$. Therefore, $d(s_1, v)\leq d(s_0, v)+\Delta_1 <   d(s_0, v)+\Delta_2$, thus $(s_1, \Delta_2)$ is also inside $X(v) \cap Y$, which is a contradiction. 

Define $v_{ij}$ to be the vertex such that $\{(s_i, \Delta_i), (s_j, \Delta_j)\}=X(v_{ij})\cap Y$. 
We construct a path $P_{ij}$ which connects from $s_i$ to $v_{ij}$ via a shortest path $P(s_i, v_{ij})$ and from $v_{ij}$ to $s_j$ via another shortest path $P(s_j, v_{ij})$. 
Now consider the five vertices $s_1, s_2, \cdots, s_5$; the paths $\{P_{ij} : \forall 1\leq i, j\leq 5\}$ topologically form a complete graph $K_5$ which is not planar. 
We will argue a contradiction in nearly the same way as in the proof of \Cref{thm:distance-VC}. Suppose paths $P_{ab}$, $P_{cd}$ intersect. It must be that one of the shortest paths,  $P(s_a, v_{ab})$ and $P(s_b, v_{ab})$, intersects one of the shortest paths $P(s_c, v_{cd})$, $P(s_d, v_{cd})$. Without loss of generality, let's assume that $P(v_{ab}, s_b)$ and path $P(v_{cd}, s_c)$ intersect.
By Lemma~\ref{lem:intersect-path}, if $v_{ab}, s_b, v_{cd}, s_c$ are pairwise disjoint, either $d(v_{ab}, s_c)\leq d(v_{ab}, s_b)\leq \Delta_b+d(v_{ab}, s_0)$
or $d(v_{cd}, s_b)\leq d(v_{cd}, s_c)\leq \Delta_c+d(v_{cd}, s_0)$. This means that either $(s_c, \Delta_b)$ is in the set $X(v_{ab})$ or $(s_b, \Delta_c)$ is in the set $X(v_{cd})$. This leads to a contradiction. Again we will need to handle the boundary cases when some pairs in $v_{ab}, s_b, v_{cd}, s_c$  intersect each other, this part is the same as in the proof of \Cref{thm:distance-VC}.
\end{proof}

\subsection{Unit-Disk and Pseudo-Disk Graphs with Dummy Vertices}

Let $G$ be an (unweighted) intersection graph of pseudo-disks. Let $\mathcal{C}$ be a set of \emph{vertex-disjoint} cliques in~$G$. For every clique $C\in \mathcal{C}$, we add a \EMPH{dummy vertex $s_C$}, remove all edges of $C$, and add edges of weight $1/2$ from $s_C$  to all vertices in $s_C$. In other words, we replace $C$ by a star centered at $s_C$ where all the edges of the star have weight $1/2$. 
Let \EMPH{$G_{\mathcal{C}}$} be the \EMPH{$\mathcal{C}$-dummy graph} obtained from $G$ by applying the star replacement to every clique in $\mathcal{C}$. 
Note that $G_{\mathcal{C}}$ is an edge-weighted graph where an edge has weight $1$ if it is in $G$ and weight $1/2$ otherwise. 
We are interested in a dummy graph because it will show up when we compute a $+1$-approximation of the diameter in \Cref{sec:plusone}. Here we  show that the VC dimension results in \Cref{thm:distance-VC} and \Cref{thm:LW-dim} for $G$ also hold for the dummy graph. 

\begin{theorem}\label{thm:dummy-graph-VC} Let $G$ be a pseudo-disk graph and $\mathcal{C}$ be a set of vertex-disjoint cliques of $G$. Let $G_{\mathcal{C}}$  be the $\mathcal{C}$-dummy graph obtained from $G$. Then:
\begin{itemize}
    \item The VC-dimension of $\mathcal{B}(G_{\mathcal{C}})$ is at most 4.
    \item Let $S$ be as any set of vertices of $G_{\mathcal{C}}$ and $M \subseteq \reals$ be any set of real numbers. Then $L_{G_{\mathcal{C}},M}(S)$ has VC-dimension at most $4$.  
\end{itemize}
\end{theorem}
\begin{proof}
We focus on proving the first result, which is bounding the VC-dimension of  $\mathcal{B}(G_{\mathcal{G}})$. The proof for the second result follows the same line. 
  
  We assume as before that there are five dummy vertices $a, b, c, d, e$ that can be shattered. Then we try to argue contradiction. We use almost the same argument as in \Cref{thm:distance-VC}. First, there is a ball centered on a vertex $v_{ab}$ with radius $r_{ab}$ includes $a, b$ but not $c, d, e$.  We take a path $P(a, b)$ as the shortest path $P(a, v_{ab})$ from $a$ to $v_{ab}$ and $P(v_{ab}, b)$ from $v_{ab}$ to $b$. Since $a, b$ are dummy nodes, the paths $P(a, v_{ab})$ can be realized by a shortest path connecting a vertex $a'$ in the clique represented by $a$ and a dummy edge connecting $a'$ to $a$ of length $1/2$. Further, we can find a representative point for $a'$ inside the pseudo-disk of $a'$, which is on the boundary of the next pseudo-disk of $a'$ along $P(a, b)$. We define the same for all pairs of vertices of $\{a, b, c, d, e\}$. These paths together with the five dummy nodes form a graph $K_5$, which has five dummy stars and ${5 \choose 2}$ paths that are also realized as geometric curves in the plane. We take $P(a', b')$ to be the subpath of $P(a, b)$ without the two dummy edges $aa'$ and $bb'$.  If two of such paths intersect (i.e., no dummy vertices/edges/cliques are involved), we use the same argument as in the proof of \Cref{thm:distance-VC}.

\begin{figure}
\centering
\includegraphics[width=0.85\linewidth]{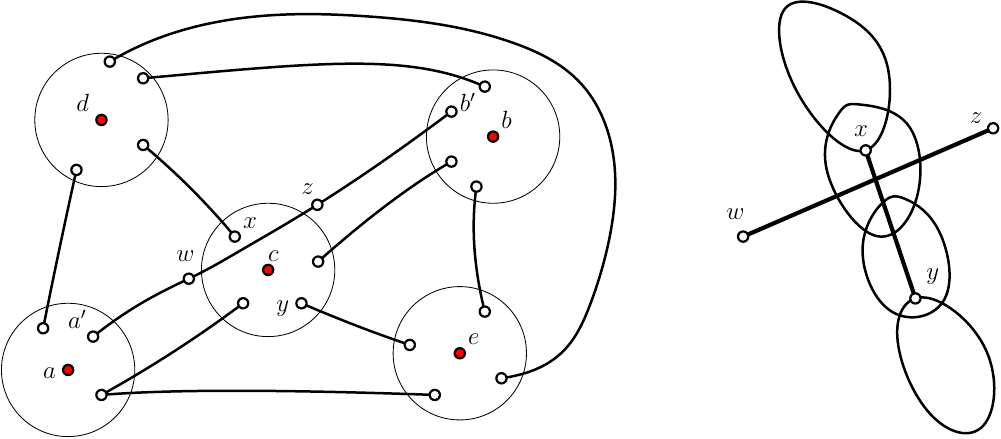}
\caption{Proof of distance VC-dimension of pseudo-disk graphs. }
\label{fig:dummy}
\end{figure}

Now we assume otherwise.  Recall that the cliques are disjoint --- no two cliques have the same vertex.  First, if any path say $P(a', b')$ involves a vertex $x$ in a different clique $c$, we immediately have a contradiction, since $d(c, v_{ab}) \leq d(a, v_{ab})$. 
The only possibility left is that one of the path $P(a', b')$ must ``go through'' another clique say the one represented by $c$, does not have any vertex of the clique $c$ on the path, and separates the vertices of the clique $c$ on both sides.
See \Cref{fig:dummy}  for an example. Specifically, we have a pair of vertices $x, y$ in the clique of $c$, staying on different side of path $P(a', b')$. Specifically, there is a pair of vertices $w, z$ on this path, none of them in the clique of $c$, and that the proper curve connecting the representative points of $w, z$ intersects (an odd number of times) with the proper curve connecting $x, y$. By the same operation as in \Cref{thm:distance-VC}, the representative points $x, y, w, z$ all stay inside two pseudo-disks. Now there would be an easy case study similar as before. Either none of $x, y, w, z$  stays inside another pseudo-disk of $x, y, w, z$---in which case we have a local crossing pattern due to \Cref{lem:intersect-pseudo}---or one of $x, y, w, z$  stays within the other pseudo-disks $x, y, w, z$, but this again gives a local crossing pattern: suppose $w$ stays within pseudo-disk $x$, then $x$ has edges to $y, w$,
as well as the upstream neighbor (towards $a$) of $w$ on path $P(a, b)$. Whenever we identify a local crossing pattern, the rest of the argument follows the same as \Cref{thm:distance-VC}.
\end{proof}

\section{\texttt{+}2-Approximation for Diameter in Unit-Disk Graphs}\label{sec:diameter-UDG}

As we discussed in \Cref{subsec:ideas}, we will combine the distance encoding with the clique-based $r$-division in \Cref{lm:clique-division}, along the line of Le and Wulff-Nilsen~\cite{lw2024}. 

\paragraph{Distance encoding.} 
Fix a sequence of vertices $S = \langle s_0,s_1,\ldots, s_{k-1}\rangle$. Following previous work~\cite{FMW21,lw2024}, we define a \EMPH{pattern} of $v$ with respect to $S$, denoted by \EMPH{$\patt_v$}, to be a $k$-dimensional vector where:
\begin{equation}\label{eq:pattern-def}
    \EMPH{$\patt_v[i]$} \coloneqq d_G(v,s_i)- d_G(v,s_0) \quad\text{for every $0\leq i \leq k-1$.}
\end{equation}
Note that $\patt_v[0] = 0$ by definition.  The following lemma is from~\cite{lw2024,FMW21}.

\begin{lemma}[\cite{lw2024,FMW21}]
\label{lm:pattern-bound-undir} 
Let $G$ be an unweighted graph and $S$ be a sequence of vertices. Suppose that $\ecd$ has VC-dimension at most $d$ for any $M$. Let $\EMPH{$\Patt$} \coloneqq \{\patt_v: v\in V\}$ be the set of all patterns with respect to $S$. Suppose that $d_G(s_i,s_0)\leq \Delta$ for every $i\in [k-1]$, then $|\Patt| = O((k\cdot \Delta)^{d})$. 
\end{lemma}

\paragraph{Computing approximate diameter.}
Our algorithm will be based on a clique-based $r$-clustering. Let $(\mathcal{R},\mathcal{C})$ be a clique-based $r$-clustering for a parameter $r$ to be chosen later. Recall that for each cluster $R\in \mathcal{R}$ the boundary vertices $\bdry R$ belong to $O(r)$ cliques in $\mathcal{C}$. 

For subgraph $R\in \mathcal{R}$, we define a sequence $\EMPH{$S_R$}$ $= \langle s_0,s_1,\ldots, s_{k_R} \rangle$ from all the clique representatives in $\rep(R)$. Note that $k_R = O(r)$ and $\rep(R) \subseteq \bdry R$ by \Cref{def:clique-r-division-2nd}. Since $R$ is connected,  by the triangle inequality, $d_G(s_i,s_0)\leq |V(R^{\circ})| + |\rep(R)| + 2  = O(r)$. For each vertex $u\in V$, we form a pattern $\patt_u$ with respect to $S_R$, and let $\EMPH{$\Patt_R$} \coloneqq \set{\patt_u : u \in V}$.  

Given a pattern $\patt_u$ of $u$ and a vertex $v\in G$, we want to estimate the distance $d_G(u,v)$ via $\patt_u$. 
This leads to the definition of distance $d(\patt,v)$ between a pattern $\patt$ and a vertex $v$:
\begin{equation}\label{eq:dis-patt-v}
    \EMPH{$d(\patt,v)$} \coloneqq \min_{i} \Set{\big. d_G(v,s_i) + \patt[i]}
\end{equation} 

Previous work~\cite{FMW21,lw2024} showed that if the sequence $S_R$  contains \emph{all} vertices of $\bdry R$, then $d_G(u,v) = d(\patt_u,v) + d_G(u,s_0)$. 
In our setting $S_R$ only contains a \emph{subset} of vertices of $\bdry R$, so recording $d(\patt,v)$ does not give us exact distances; 
however, we get a $+2$-approximation as shown by the following lemma.

\begin{lemma}
\label{lm:dist-via-pattern} 
Suppose that $\pi(u,v,G)\cap \bdry R\not=\varnothing$ where $\pi(u,v,G)$ is a shortest path from $u$ to $v$ in $G$. Let
\begin{equation}
\label{eq:approx-dist}
    \EMPH{$\tilde{d}_G(u,v)$} \coloneqq d_G(u,s_0) + d(\patt_u,v).
\end{equation}
Then, $d_G(u,v)\leq \tilde{d}_G(u,v)\leq d_G(u,v)+2$. 
\end{lemma} 
\begin{proof}
By definition, $d_G(u,s_{i}) = \patt_u[i] + d_G(u,s_0)$ holds for each boundary vertex $s_{i}$ where $0\leq i\leq  k_R$. First, observe that:
 \begin{equation}\label{eq:dtilde-uv-expand}
 \begin{split}
     \tilde{d}_G(u,v) &= d_G(u,s_0) +   \min_{0\leq i \leq  k_R} \Set{\Big.  d_G(v,s_i) + \patt_u[i] }\\
     &=   \min_{0\leq i \leq  k_R} \Set{\Big.  d_G(u,s_0) +  d_G(v,s_i) + \patt_u[i] }\\
     &= \min_{0\leq i \leq  k_R} \Set{\Big.  d_G(u,s_i) +  d_G(v,s_i) }
\end{split}
\end{equation}
Thus, $\tilde{d}_G(u,v) \geq d_G(u,v)$ holds by triangle inequality. 

For the other direction, let $x\in \pi(u,v,G)\cap \bdry R$; $x$ exists by the assumption of the lemma. Let $s_{x}$ be the boundary vertex in $S_R$ that is in the same clique with $x$.  
By \Cref{eq:dtilde-uv-expand} and the triangle inequality:
\[
   \tilde{d}_G(u,v)\leq   d_G(u,s_x) + d_G(s_x,v) \leq d_G(u,x) + d_G(x,v) + 2 = d_G(u,v)+2, 
\]
as desired.
\end{proof}

 We now describe our algorithm. 
The \EMPH{eccentricity} of a vertex $u$ is defined to be $\EMPH{$\ecc(u)$} \coloneqq \max_{v\in V(G)}d_G(u,v)$.  
We will compute the approximate diameter by computing the approximate eccentricities for all vertices in $G$; that is, for each vertex $u$, we will compute an approximation $\widetilde{\ecc}(u)$, and then output $\max_{u}\widetilde{\ecc}(u)$.  Our algorithm is similar to the algorithm of Le and Wullf-Nilsen~\cite{lw2024} for computing exact diameter in minor-free graphs. Here, we use clique-based $r$-clustering in place of an $r$-division and have to handle the cliques in $\mathcal{C}$. We also have to be more careful in the way we handle clusters in $\mathcal{R}$ as a cluster could have a very large size. The algorithm has three steps:

\begin{itemize}
\item \textsf{Step 1.~} Construct a clique-based $r$-clustering $(\mathcal{R},\mathcal{C})$ of $G$. For each clique $K(x)$ in $\mathcal{C}$ represented by a vertex $x$, we find the shortest path distances from $x$ to all other vertices of $G$ using a single-source shortest path algorithm~\cite{Efrat2001-hm}.  For each cluster $R\in \mathcal{R}$, form a sequence of boundary vertices $S_R$ as described above.  We compute a set of patterns $\Patt_R = \{u\in V: \patt_u\}$ with respect to $S_R$. We store all the information computed in this step in a table \EMPH{$T^{(1)}_R$}.

\item \textsf{Step 2.~} For each cluster $R\in \mathcal{R}$, each pattern $\patt \in \Patt_R$, and each vertex $v\in R^{\circ}$, we compute $d(v,\patt)$. 
Then we find $\EMPH{$v_{\patt}$} \coloneqq \argmax_{v\in V(R^{\circ})} d(v,\patt)$, which is the \emph{furthest vertex} from $\patt$.  
We store both $d(v,\patt)$ and $d(\patt,v_{\patt})$ in a table  \EMPH{$T^{(2)}_R$}.

\item \textsf{Step 3.~} We now compute $\widetilde{\ecc}(u)$ for each vertex $u\in V$. For each cluster $R\in \mathcal{R}$, we compute the approximate distance from $u$ to the vertex $v\in R^{\circ}$ furthest from $u$, denoted by  \EMPH{$\Delta(u,R^{\circ})$}, as follows. 
Let $\patt_u$ be the pattern of $u$ with respect to $S_R$ computed in Step~1.
\begin{itemize}
    \item \textsf{Step 3a.~} If $u\not\in R^{\circ}$,  let $v$ be the furthest vertex from $\patt_u$, computed in Step 2. We return $\Delta(u,R^{\circ}) \coloneqq d_G(u,s_0) + d(\patt_u,v)$ where $s_0$ is the first vertex of $S_R$. 
    
    \item  \textsf{Step 3b.~}  If $u\in R^\circ$, then we compute a distance \EMPH{$d_{R^{\circ}}(u,v)$} from $u$ to every $v\in R^{\circ}$  where this distance is in the intersection graph of the disks in $R^{\circ}$.   Then, compute $\tilde{d}_G(u,v) = d_G(u,s_0) + d(\patt_u,v)$ and finally return $\Delta(u,R^{\circ}) \coloneqq \max_{v\in R^{\circ}} \min\Set{\big. \tilde{d}_G(u,v), d_{R^{\circ}}(u,v)}$.
\end{itemize}
We are not done yet: we have to compute the maximum approximate distance, denoted by $\Delta(u,\mathcal{C})$ from $u$ to vertices in cliques in $\mathcal{C}$: 
\begin{equation}\label{eq:delta-u-C}
    \Delta(u,\mathcal{C}) = 1 + \max_{K(x)\in \mathcal{C}} d_G(u,x)
\end{equation}
Here $K(x)$ is the clique in $\mathcal{C}$ represented by $x$.  The distance $ d_G(u,x)$ was computed and store in $T^{(1)}_R$ in Step 1. Finally, we compute:
\[    \widetilde{\ecc}(u) = \max \left\{\max_{R\in \mathcal{R}}\Delta(u,R^{\circ}), \Delta(u,\mathcal{C})\right\}~.
\]
\end{itemize}

\paragraph{Correctness.}  
Let $v^*$ be the furthest vertex from $u$; that is, $d_G(u,v^*) = \ecc(u)$. 
If $v^*$ belongs to some clique $K(x) \in \mathcal{C}$, then by the triangle inequality, $d_G(u,x)-1\leq d_G(u,v^*)\leq d_G(u,x)+1$, implying that  $d_G(u,v^*)\leq d_G(u,x)+1 \leq d_G(u,v^*) + 2$. 
Thus, $\Delta(u,\mathcal{C})$ computed in \Cref{eq:delta-u-C} is a $+2$ approximation of $d_G(u,v^*)$, and hence $\ecc(u)$, in this case. 

Otherwise, $v^*$ does not belong to some clique $K(x) \in \mathcal{C}$. By the definition of $r$-clustering, Item 4, $v^*\in R^{\circ}$ for some cluster $R\in \mathcal{R}$.  If $u\not\in R^{\circ}$, then  $\pi(u,v^*,G)\cap \bdry R\not=\varnothing$ and hence $\Delta(u,R^{\circ})$ is a $+2$-approximation of $\max_{v\in R^{\circ}} d_G(u,v) =  d_G(u,v^*)$ by \Cref{lm:dist-via-pattern}. If $u\in R^{\circ}$, the algorithm has to account for the fact that $\pi(u,v^*,G)$ could contain vertices outside $R$. If $\pi(u,v^*,G)\cap \bdry R\not=\varnothing$, then $d_G(u,s_0) + d(\patt_u,v^*)$ is a $+2$-approximation of $d_G(u,v^*)$ by \Cref{lm:dist-via-pattern}. Otherwise, $\pi(u,v^*,G)\subseteq R^{\circ}$ and hence $d_{R^{\circ}}(u,v^*) = d_G(u,v^*)$. Considering both cases, we conclude that $\tilde{d}_G(u,v^*)$ is a $+2$-approximation of $d_G(u,v^*)$ and hence $\Delta(u,R^{\circ})$ is a $+2$-approximation of $\tilde{d}_G(u,v^*)$. 

In both cases, we have $\ecc(u)\leq \widetilde{\ecc}(u) \leq \ecc(u)+2$, implying that the algorithm returns a $+2$-approximation of the diameter.

\paragraph{Running time.}   
In Step 1, we compute the shortest distances from representatives of cliques in $\mathcal{C}$ to all other vertices. As $|\mathcal{C}| = O(n/\sqrt{r})$ and finding single-source shortest paths in unit-disk graphs can be done in $O(n\log n)$ time~\cite{Efrat2001-hm,Cabello2015-vo,Chan2016-sy}, the running time to compute all these distances is $\Tilde{O}(n^2/\sqrt{r})$. Then for each vertex $u\in V$, and for each $R\in \mathcal{R}$, computing $\patt_u$ can be done by looking up the distances from the representatives in $\rep(R)$.  The running time is $O(|\rep(R)|) = O(|\mathcal{C}(\bdry R)|)$ for each $u$ and $R$. Thus, the total running time is:
\begin{equation*}
\begin{split}
    \sum_{u\in V} \sum_{R\in \mathcal{R}}  O(|\mathcal{C}(\bdry R)|) &= \sum_{u\in V} O(n/\sqrt{r}) \qquad \text{(by Item 3 in \Cref{def:clique-r-division-2nd})}\\
    &= O(n^2/\sqrt{r}).     
\end{split}
\end{equation*}
Therefore, Step 1 could be implemented in $\Tilde{O}(n^2/\sqrt{r})$ time.  

Next for Step 2, by \Cref{lm:pattern-bound-undir}, the number of patterns $|\Patt_R| = r^{2d}$. (Note that $d = 4$ by \Cref{thm:LW-dim}.) 
Thus, we can implement Step 2 in $O(\frac{n}{\sqrt{r}} \cdot r^{2d} \cdot r) = O(nr^{(4d+1)/2})$ time.

Finally, we account for the running time in Step 3. 
Step 3a could be done in $O(1)$ time per vertex $u$ and cluster $R$. 
Thus, the total running time is $O(n|\mathcal{R}|) = O(n^2/\sqrt{r})$. For Step 3b, we only restrict to $u$ belongs to $R^\circ$ and there are only $r$ such vertices. 
Computing $\Delta(u,R^\circ)$ in this case can be done in $\Tilde{O}(r)$ time per vertex $u$, using single-shortest paths in unit-disk graphs~\cite{Efrat2001-hm}. 
Thus, the total running time of Step 3b over all $u\in R^{\circ}$ and all $R\in \mathcal{R}$ is $\Tilde{O}(|\mathcal{R}| \cdot r^2) = \Tilde{O}(nr^{3/2})$. Lastly, computing $\Delta(u,\mathcal{C})$ \emph{for all $u$} can be done in time $O(n|\mathcal{C}|) = O(n^2/\sqrt{r})$ by looking up the distances computed from Step 1. Thus, the total running time of Step 3 is $O(n^2/\sqrt{r} + n\cdot r^{3/2})$. 

In summary, the total running time of the entire algorithm is
\begin{equation}
    \Tilde{O}\Paren{ \frac{n^2}{\sqrt{r}} + n r^{(4d+1)/2}} = \Tilde{O}(n^{2-\frac{1}{4d+2}}) = \Tilde{O}(n^{2 - 1/18})
\end{equation}
by setting $r = n^{1/(2d+1)} = n^{1/9}$.

\section{\texttt{+}2-Approximation Distance Oracle for Unit-Disk Graphs}\label{sec:oracle-UDG}

Similar to Section~\ref{sec:diameter-UDG}, we now show how to construct distance oracle on unit-disk graphs with merely +2 error.
First we describe the construction of the distance oracle.

\begin{itemize}
\item \textsf{Step 1.~}  
Construct a clique-based $r$-clustering $\mathcal{R}$ of $G$ using Lemma~\ref{lm:clique-division}.  
For each subgraph $R\in \mathcal{R}$, form a sequence of $S_R \coloneqq \Seq{s_0,\dots,s_{k_R}}$ from all clique representatives in $\rep(R)$.
We have $S_R = \abs{O(r)}$ by \Cref{def:clique-r-division-2nd}.
We compute a set of patterns $\Patt_R \coloneqq \set{ \patt_u : u\in V  }$ with respect to $S_R$, and store it in a table \EMPH{$T^{(1)}_R$}.

\item \textsf{Step 2.~} For each subgraph $R\in \mathcal{R}$ and each vertex $v$: 
(a) if $v\in R^\circ$ or $v$ is a representative in $S_R$, we compute and store $d(v,\patt)$ for each pattern $\patt \in \Patt_R$;
(b) if $v \not\in R^\circ$, we find the pattern $\patt_v$ of $v$ with respect to $S_R$, and store a pointer from $v$ to $\patt_v$ and the distance $d_G(v,s_0)$ in a table \EMPH{$T^{(2)}_R$}.

\item \textsf{Step 3.~} For each subgraph $R\in \mathcal{R}$, compute $d_G(u,v)$ for every pair of vertices $(u,v)$ in $R^\circ$, and store them in a table \EMPH{$T^{(3)}_R$}.
\end{itemize}

\noindent For any distance query between vertices $u$ and $v$, we perform the following.
\begin{itemize}
\item
If there is a subgraph $R\in \mathcal{R}$ such that $R^\circ$ containing both $u$ and $v$, then we can return their distance in $G$ using table $T^{(3)}_R$.

\item
Otherwise, let $R$ be the subgraph containing $v$. 
We compute the approximate distance from $u$ to $v$ as follows. 
Let $\patt_u$ be the pattern of $u$ with respect to $S_R$ computed in Step~1.
First we look up the distance between $u$ and $s_0$ and the pointer from $u$ to $\patt_u$ from table $T^{(2)}_R$. 
If $v$ belongs to some clique in $\bdry R$ with representative point $x$, then we look up the distance the distance $d(\patt_u, x)$ from table $T^{(2)}_R$.  
And we return $\tilde{d}_G(u,v) \coloneqq d_G(u,s_0) + d(\patt_u, x)$.

\item
Else, $v$ belongs to none of the cliques in $\bdry R$, and by definition of $r$-clustering, $v$ must be in $R^\circ$.
Then we look up  the distance $d(\patt_u, v)$ again from table $T^{(2)}_R$.  
Finally we return $\tilde{d}_G(u,v) \coloneqq d_G(u,s_0) + d(\patt_u, v)$.
\end{itemize}

\paragraph{Analysis.}
The correctness of the construction again follows from Lemma~\ref{lm:dist-via-pattern}.
Querying the distance between a given pair of vertices $(u,v)$ takes $O(1)$ time as every necessary information are stored in the tables.
As for space analysis, the number of patterns is $r^{2d}$ by \Cref{lm:pattern-bound-undir}.  (Here $d=4$ by \Cref{thm:LW-dim}.)
Table $T^{(1)}_R$ takes $O(\frac{n}{\sqrt{r}} \cdot r^{2d} \cdot r) = O(n r^{2d+1/2})$ space.
Table $T^{(2)}_R$ takes $O(\frac{n}{\sqrt{r}} \cdot (r\cdot r^{2d} + n)) = O(n r^{2d+1/2} + \frac{n^2}{\sqrt{r}})$ space.
Table $T^{(3)}_R$ takes $O(\frac{n}{r} \cdot r^{2}) = O(n r)$ space.
Thus in total the distance oracle uses $O(n r^{2d+1/2} + \frac{n^2}{\sqrt{r}})$ space.   
Taking $r = n^{1/(2d+1)} = n^{1/9}$ gives us $O(n^{2-1/18})$ space.

\section{Well-Separated Clique-Based Separator Decomposition}
\label{sec:clique-division}

In this section, we prove Lemma~\ref{lm:clique-division}; see \Cref{subsec:ideas} for an overview of the argument.  

\begin{definition}[Well-separated clique-based separators]\label{def:well-sep} 
Let $D$ be a set of $n$ unit-disks. 
Let $G$ be its geometric intersection graph.   We say a family of disjoint subsets of cliques of $G$, denoted by \EMPH{$\mathcal{S}$}, is a \EMPH{well-separated clique-based separator} of $D$ if all following conditions hold:
\begin{itemize}
    \item  {\normalfont [Balanced.]}  Every connected component of $G\setminus \mathcal{S}$ contains at most $2n/3$ disks.  
    \item  {\normalfont [Well-separated.]}  For every two disks $a$ and $b$ in different components of $G\setminus \mathcal{S}$, the minimum Euclidean distance $\Edist{a}{b}$ between points in $a$ and points in~$b$ is greater than $2$. 
    \item  {\normalfont [Low-ply.]}  The disks in each clique in $\mathcal{S}$ are stabbed by a single point. Furthermore, we could choose for each clique $X\in \mathcal{S}$ a \emph{representative disk} $x$ such that the ply with respect to the intersection graph of all representative disks in $\mathcal{S}$ is  $O(1)$. 
\end{itemize}
We say that the \EMPH{size} of $\mathcal{S}$ is the number of cliques in $\mathcal{S}$.
\end{definition}

We will show that by adapting the clique-based separator theorem for geometric intersection graphs~\cite{BBKMv20,deBerg23}, 
we can construct a well-separated clique-based separator for unit-disk graphs in near-linear time. 
The proof of the following lemma can be found in Section~\ref{subsec:well-separated-separator}.

\begin{lemma}
\label{lm:disk-well-sep} Let $D$ be a set of $n$ unit disks. 
We can construct a well-separated clique-based separator $\mathcal{S}$ for $D$ of size $O(\sqrt{n})$ in $O(n\log n)$ time,
such that for every disk $y$, there are $O(1)$ cliques in $\mathcal{S}$ that intersect $y$. 
Furthermore, we can compute the list of the representative disks of cliques in $\mathcal{S}$ that intersect $y$ for every disk $y$ in a total of $O(n\log n)$ time.  
\end{lemma}

\paragraph{Clique-based \boldmath{$r$}-clustering algorithm.~} Let $D$ be the set of $n$ disks and $G$ be its geometric intersection graph. In this step, we recursively partition $D$ into a family of sets of disks such that each set has at most $r$ disks and at most $r$  boundary cliques.  We also maintain a (global) set of cliques $\mathcal{C}$ and their representative disks. For each representative disk $x$, let \EMPH{$K(x)$} be the clique in $\mathcal{C}$ represented by $x$. 

At each intermediate recursive step, we will maintain an (explicit) set \EMPH{$\hat{D}$} of size at least $r$ that includes two types of disks:~\EMPH{regular disks} and  \EMPH{representative disks}. We assume that $|\hat{D}|\geq r$; otherwise, the algorithm will stop in the previous step.  
For each regular disk $y \in \hat{D}$, we maintain a list of representative disks, denoted by \EMPH{$\rho(y)$},  in $\hat{D}$ such that for each $x \in \rho(y)$ the clique $K(x)$ it represents has at least one disk that intersects with $y$.  (Notice that the representative $x$ itself might not intersect $y$.) 
Furthermore, we will show below (\Cref{clm:boundary-T}) that every neighbor of $y$ in $G$ is in a clique represented by some disk in $\rho(y)$. 
Let \EMPH{$\Gamma(\hat{D})$} be the graph obtained from $\hat{D}$ by first taking the intersection graph of $\hat{D}$ and, for every regular disk $y$, adding an edge $(y,x)$ for every $x\in \rho(y)$.  (Intuitively, we pretend as if the representative $x$ itself intersects $y$ instead of the clique $K(x)$.)
We call $\Gamma(\hat{D})$  the \EMPH{extended intersection graph} of $\hat{D}$. We will ensure that $\Gamma(\hat{D})$  is a connected graph. 
Note that we will not explicitly maintain $\Gamma(\hat{D})$ as it could have super-linear many edges, where as our goal is near-linear time. 
Initially, $\hat{D} = D$ and $\Gamma(\hat{D}) = G$, and all disks in $\hat{D}$ are regular disks. 

\begin{figure}[htbp]
\centering
\includegraphics[width=0.4\linewidth]{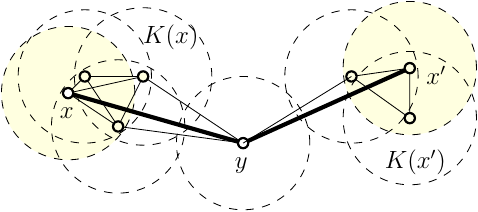}
\caption{The extended intersection graph $\Gamma(\hat{D})$. 
The edges connecting a regular disk $y$ and representative disks in $\rho(y)$ are shown in solid edges.}
\label{fig:extended-intersection-graph}
\end{figure}

We then apply \Cref{lm:disk-well-sep} to construct a well-separated clique-based separator \EMPH{$\hat{\mathcal{S}}$} for $\hat{D}$. If there is a representative disk $x$ in $\hat{D}$ contained in a clique in $\hat{\mathcal{S}}$, we split $x$ out of the clique and consider $x$ an independent clique in $\hat{\mathcal{S}}$. We then add new cliques in  $\hat{\mathcal{S}}$ to $\mathcal{C}$. 

Let \EMPH{$R(\hat{\mathcal{S}})$} be the set of representative disks of cliques in $\hat{\mathcal{S}}$.  
We partition $\hat{D}\setminus \hat{\mathcal{S}}$ into two set of disks $D_1,D_2$ each contains at most $2|\hat{D}|/3$ disks. 
For each $D_i$, we construct a spanning forest $F$ of the extended intersection graph $\Gamma(D_i\cup R(\hat{\mathcal{S}}))$. 
For each connected component $T$ of $F$, if $T$ has at most $r$ vertices, we form a cluster 
\EMPH{$R_T$} containing all regular and representative disks of $T$, and all disks in the clique of the representative disks of $T$, and add $R_T$ to $\mathcal{R}$.  Otherwise, $T$ has at least $r$ vertices, we recurse on the set of disks, say $\hat{D}_i$, corresponding to vertices of $T$. 
The extended intersection graph of $\hat{D}_i$ will be connected. 

\paragraph{Running time analysis.} First we bound the number of disks, denoted by $A(n)$, counted with multiplicity, over the course of the algorithms; these are disks in $\hat{D}$ for every $\hat{D}$ that appeared in the recursion. Observe that $A(n)$ satisfies: 
 \begin{equation}
     A(n) \leq A(n_1) + A(n_2) + O(n)
 \end{equation}
 for $n_1$ and $n_2$ such that $n_1 + n_2 = n + O(\sqrt{n})$ and $n_1,n_2\geq n/3$. By induction, we could show that $A(n) = O(n\log n)$. Therefore, to show that the total running time is $O(n \log^2 n)$, it suffices to show that in each recursion step, the total running time is $O(|\hat{D}|\log n)$. 

Let $\hat{n} \coloneqq |\hat{D}|$.  Observe that constructing $\hat{\mathcal{S}}$ for $\hat{D}$ takes $O(\hat{n}\log \hat{n})$ time. 
Then, for $D_1$ (as well as $D_2$), we construct a spanning forest $F$ of $\Gamma(D_1\cup R(\hat{\mathcal{S}}))$. Note that each regular disk $y$ maintains a list of representative disks $\rho(y)$, and by \Cref{lm:disk-well-sep}, $|\rho(y)| = O(\log n)$ as the recursion depth is $O(\log n)$. Then, to compute  $F_1$, we simply compute a spanning forest for the intersection graph of regular disks, which could be done in $O(\hat{n}\log \hat{n})$ by computing the Delaunay triangulation of the centers of these disks, and then add edges to the representative disks. The total running time is $O(\hat{n}\log n)$, as desired.

\paragraph{Bounding boundary cliques of \boldmath{$\mathcal{R}$}.} 
The same argument of Frederickson (proof of Lemma 1 in~\cite{Frederickson1987}) applies to bound $\sum_{R\in \mathcal{R}}|\mathcal{C}(\partial\! R)|$, which is the number of cliques in $\mathcal{C}$ generated by the $r$-clustering algorithm, counted with multiplicity.  
We reproduce Frederickson's argument here almost verbatim for completeness. Let $B(n,r)$ be the number of cliques in $\mathcal{C}$ counted with multiplicities.  Then we have:
 \begin{equation}
 \begin{split}
     B(n,r) &\leq c_0 \sqrt{n} + B(\alpha n + O(\sqrt{n}),r) +  B((1-\alpha) n + O(\sqrt{n}),r) \qquad \text{for $n >  r$}\\
     B(n,r) &= 0 \qquad \text{for $n \leq   r$}
 \end{split}
 \end{equation}
for some $\alpha \in [1/3,2/3]$. Then, by induction $B(n,r)\leq c_0 n/\sqrt{r} - d \sqrt{n}$ for a sufficiently large constant $d$, implying Item (3) of Definition~\ref{def:clique-r-division-2nd}.

\paragraph{Analyzing properties of \boldmath{$\mathcal{R}$}.} 
Recall that each cluster $R_T$ formed from a spanning $T$ and some disks in the cliques of the representative disks in $T$. 
Thus, $R_T$ induced a connected subgraph of $G$, the geometric intersection graph of $D$. 
Furthermore, $T$ has at least one representative disk based on the way we constructed the extended intersection graph $\Gamma(D_1\cup R(\hat{\mathcal{S}}))$, so the number of clusters in $\mathcal{R}$ is at most the number of cliques in $\mathcal{C}$ counted with multiplicity, which is $O(n/\sqrt{r})$ as shown above. 
This implies Item (1) in Definition~\ref{def:clique-r-division-2nd}. 
To show Item (2), we claim that:

\begin{claim}\label{clm:boundary-T} For every regular disk $y$ in a tree $T\in F$, any neighbor (in $G$) of $y$ not in $T$ belongs to some clique represented by representative disks in $T$. 
\end{claim}
\begin{proof}
    We prove the claim by induction on the recursive steps of the algorithm.   
    Given the current set of disks $\hat{D}$, we inductively assume that any neighbor of $y$ (in $G$) is either in $\hat{D}$ or in a clique $K(x)$ of some representative disk $x$ in $\hat{D}$. 
    Without loss of generality, we assume that $y\in D_1$ and let $z$ be a neighbor in $G$ of $y$ that is not in $T$. 
    Then edge $yz$ is not in $T$ and thus $z$ cannot be in $\hat{D}$; by induction hypothesis we have $z\in K(x)$. 
    If $x\in \hat{\mathcal{S}}$, we are done, since we consider all representative disks of cliques in $\hat{\mathcal{S}}$ in the construction of $F$. 
    Otherwise, we show that $x$ must belong to $D_1$: otherwise $x$ is in $D_2$ as it is not in $\hat{\mathcal{S}}$. 
    However, this contradicts the well-separated property of $\hat{\mathcal{S}}$ in \Cref{def:well-sep} since $\Edist{y}{x}\leq 2$. 
    Thus, $x\in D_1$, and hence $y$ has an edge to $x$ in $\Gamma(D_1\cup R(\hat{\mathcal{S}}))$, meaning that $y$ must belongs to the connected component of $x$ in     $\Gamma(D_1\cup R(\hat{\mathcal{S}}))$, which is $T$. 
\end{proof}

As $T$ has at most $r$ vertices, it has at most $r$ representative disks and regular disks. 
By \Cref{clm:boundary-T}, every disk in $\partial\! R$ belongs to a clique represented by a representative disk in $T_1$.  
Note that the regular disks of $T_1$ are those in $R^{\circ}$. Thus, Item (2) in Definition~\ref{def:clique-r-division-2nd} follows. Furthermore, by construction, every vertex of $G$ either belongs to $R^\circ$ for some cluster $R$ in $\mathcal{R}$ or a clique in $\mathcal{C}$, implying Item (4) in Definition~\ref{def:clique-r-division-2nd} and hence Lemma~\ref{lm:clique-division}.

\begin{remark}\label{remark:clique-based}  
We could generalize our algorithm for constructing a clique-based $r$-clustering for unit-disk graphs to more general cases. 
We note that \Cref{def:well-sep} applies to intersection graphs of any geometric objects, assuming that the diameter of every object is at most $1$ by scaling.  
Recall that to construct a clique-based $r$-clustering, we need 
(1) an algorithm for constructing a well-separated clique-based separator running in $\Tilde{O}(n)$ time and 
(2) an algorithm to construct a spanning tree of the geometric intersection graphs running in $\Tilde{O}(n)$ time. 
As long as we have these two components, our algorithm in this section gives a clique-based $r$-clustering with running time  $\Tilde{O}(n)$. 
\end{remark}

\subsection{Well-separated clique-based separator}\label{subsec:well-separated-separator}

It remains to prove \Cref{lm:disk-well-sep}. 
Our algorithm is a modification of the algorithm by de Berg~\cite[Theorem~2]{deBerg23}, which is an efficient implementation of the clique-based separators for geometric intersection graphs by de Berg \etal~\cite{BBKMv20}. 
We tailor their algorithm to unit-disk graphs to get a well-separated clique-based separator. The algorithm has several steps. 

 \begin{itemize}
     \item \textsf{Step 1.} Let \EMPH{$H$} be a minimal square such that the interior of $H$ contains at most $n/12$ disks. 
     By scaling, we assume that $H$ has a side length of $1$. Let \EMPH{$\bar{r}$} be the \emph{scaled radius} of the disks. 
     
     \item  \textsf{Step 2.} Consider $\sqrt{n}$ slightly bigger squares $\EMPH{$\mathcal{H}$} = \{H_i\}_{i=1}^{\sqrt{n}}$ with side-length equally spaced between $1$ and $2$, sharing the same center with $H$. More precisely, for every $i\in [1 \,..\, \sqrt{n}]$, the square $H_i$ has side length $1+\frac{i}{\sqrt{n}}$.  
     Every square $H_i$ is contained in a square $Q$ of side-length $2$. 
     Furthermore, by the minimality of $H$, at most $4 \cdot n/12 \leq n/3$ disks intersect $Q$, because $Q$ can be divided into $4$ squares of side-length $1$, thus each intersecting at most $n/12$ disks.
     
     \item \textsf{Step 3.}   If every disk has a diameter at most $1/\sqrt{n}$, this means each disk could intersect the boundary of at most two squares in $\mathcal{H}$. Thus, there exists at least one $i \in [1,\sqrt{n}-2]$ such that for 4 consecutive squares $H_i,H_{i+1}, H_{i+2}, H_{i+3}$, the number of disks completely contained inside $H_{i+3}\setminus H_i$ is at most $O(\frac{n}{\sqrt{n}-3}) = O(\sqrt{n})$ disks in total. 
     Let \EMPH{$S$} be this set of disks. We could take each clique to be a single disk in $S$, but the representative disks of resulting cliques might have unbounded ply. 
     To reduce the ply, we considered the grids of cell length $\bar{r}/2\times \bar{r}/2$ restricted to $H_{i+3}\setminus H_i$; any disk in $S$ will intersect some of the grid points. Then, we create a set of cliques \EMPH{$\mathcal{S}$} by adding disks in $S$ stabbed by the same grid point as a single clique. 
     (If a disk is stabbed by more than one grid point, then arbitrarily assign it to one of the cliques.) Observe that $\mathcal{S}$ satisfies all the properties in \Cref{def:well-sep}. Specifically, the well-separated properties follow from the fact that for any two disks $a,b$ not in $S$ such that $a$ intersects the boundary or completely outside of $H_{i+3}$ and $b$ intersects the boundary or completely inside of $H_{i}$, $\Edist{a}{b} > \Edist{H_{i+1}}{H_{i+2}}\geq \bar{r}$, which is $2$ in the unscaled distance.  Also, each disk in $D\setminus S$ could only intersects $O(1)$ cliques in $\mathcal{S}$. 

     \item  \textsf{Step 4.}  Otherwise, let $\EMPH{$r_i$} \coloneqq 2^{i}/\sqrt{n}$ be such that the diameter of the disks is in $(r_i/2,r_i]$ for some $i\in [1,\sqrt{n}]$. Then each disk could intersect the boundary of at most $O(2^{i})$ squares in $\mathcal{H}$. 
     Look at the subset $\mathcal{H}_i$ of $\mathcal{H}$ where the boundaries of these squares are equally spaced at distance $2^{i}\cdot \sqrt{n}$. 
     Then $|\mathcal{H}_i|= \Theta(\sqrt{n}/2^i)$. 
     Let \EMPH{$N_i$} be the set of grid points (of the $\bar{r}/2\times \bar{r}/2$ grid) in the big square $Q$ (of side length 2). Note that $\bar{r} \in (r_{i}/4, r_i/2]$.  
     Thus, $|N_i| = O(1/r_i^2) = O(\frac{n}{2^{2i}})$. 
     Each grid point $p\in N_i$ defines a clique of disks stabbed by $p$, and furthermore, every disk in $Q$ must stab a point in $N_i$. Thus, $|N_i|$ is the upper bound on the number of cliques in $Q$. 

    On the other hand, every disk stabbed by a point $p\in N_i$ could intersect $O(1)$ boundaries of squares in  $\mathcal{H}_i$, which is in the distance $O(1) \cdot r_i$ from the point $p$. Furthermore, the number of grid points of $N_i$ within distance $c\cdot r_i$, for any constant $c\geq 1$, from the boundary of each square in $\mathcal{H}_i$ is at most:
    \(
        O(c/r_i) =  O(\frac{\sqrt{n}}{2^i})
    \)
    points in $N_i$. 
    Thus, \emph{any square} in $\mathcal{H}_i$ defines a well-separated clique-based separator, which contains those stabbed by points in $N_i$ within distance $c\cdot r_i$, for a sufficiently large~$c$, from the boundary of $H_i$. 
    The number of cliques is $ O(\frac{\sqrt{n}}{2^i}) = O(\sqrt{n})$, as claimed in \Cref{lm:disk-well-sep}. 
 \end{itemize}

\paragraph{Running time.} 
The only difference between our algorithm and the algorithm for finding the clique-based separator for unit-disk graphs is that in Step 3, we take cliques defined by at most 4 consecutive squares instead of using only 1. Thus, our running time is the same as the running time of the implementation outlined by de Berg~\cite{deBerg23} for geometric intersection graphs, which is $O(n\log n)$.  Indeed, implementing our algorithm is much simpler as we do not have to deal with ``large objects'' and ``small objects'' separately, as every disk has the same size.

\begin{remark}
\label{rm:well-sepclique-basd} (a) For each clique in the separator, all objects in the clique could be stabbed by a single point; this fact will be helpful in the case of bounded ply. 

(b) Our algorithm to construct a well-separated clique-based separator for unit-disk graphs could be applied to construct a well-separated clique-based separator for geometric intersection graphs of \textul{similar-size objects with constant complexity}.  Specifically, for these objects, the same notion in \Cref{def:well-sep} applies, assuming that we scale the objects so that {each has a diameter at most $1$} and at least $\Omega(1)$.  
Then, in Step 3, instead of considering $4$ consecutive squares, we consider $c$ consecutive squares for a sufficiently big constant $c$. Step 4 remains unchanged, except that the constant $c$ now depends on the minimum size of the objects. The separator algorithm of de Berg~\cite{deBerg23} works for geometric intersection graphs of fat objects with constant complexity, which applies in our case. 
\end{remark}
\section{Extension to Graphs of Similar Size Pseudo-Disks}

We consider a pseudo-disk graph, where the graph is defined as the intersection graph of a set of pseudo-disks. The following algorithm works for pseudo-disks that are of roughly the same size and have constant complexity. Specifically, we assume that the pseudo-disks are fat objects that are sandwiched between two disks of the same center of radius $r$ and $R$, $r\leq R$ being two fixed constants. 
We refer to a pseudo-disk with center $p$ as $C_p$. We also assume that the boundary of each object can be represented by a constant number of algebraic arcs.

\subsection{Approximate diameter and distance oracles}

In this section, we prove \Cref{thm:pseudodisk-diam-addive} and \Cref{thm:pesodo-dks-oracle}. 
Recall that in computing a $+2$-approximation of diameter and distance oracle for unit-disk graphs in \Cref{sec:diameter-UDG} and \Cref{sec:oracle-UDG}, respectively, we used two technical ingredients: (i) a clique-based $r$-clustering computable in $\Tilde{O}(n)$ time and (ii) an algorithm for computing single-source shortest path in unit-disk graphs in $\Tilde{O}(n)$ time. As long as we have the two technical ingredients for the intersection graphs of similar-size pseudo-disks of constant complexity, we then have the algorithms for computing a $+2$ approximation of diameter and distance oracle with the same guarantees for these graphs.

In \Cref{subset:SSP}, we show how to compute the single-source shortest path for the intersection graphs of similar-size pseudo-disks in $\Tilde{O}(n)$ time; see \Cref{thm:SSSP}. 
Here, we show how clique-based $r$-clustering for intersection graphs of similar-size pseudo-disks can also be done in $\Tilde{O}(n)$ time; this will implies \Cref{thm:pseudodisk-diam-addive} and \Cref{thm:pesodo-dks-oracle}.

\begin{lemma}\label{lm:rclustering-pseudo-disk} 
For any given $r$ and an $n$-vertex intersection graph $G$ of similar-size pseudo-disk of constant complexity, we can find the implicit representation of a clique-based $r$-clustering $(\mathcal{R},\mathcal{C})$ of $G$ in $\Tilde{O}(n)$~time.
\end{lemma}

\begin{proof} 
By \Cref{rm:well-sepclique-basd}, $G$ has a well-separated clique-based separator that can be constructed in $\Tilde{O}(n)$ time.  Thus, by \Cref{remark:clique-based}, we only need to have a $\Tilde{O}(n)$ time algorithm to construct a spanning tree of the intersection graph of the pseudo-disks. Here, we could use our single-source shortest-path algorithm in \Cref{subset:SSP} to find a spanning tree; this implies the lemma.
\end{proof}

\subsection{Exact diameter for small ply}

In this section, we prove \Cref{thm:pseudodisk-diam-exact}. Observe that when the objects have ply $k$, by \Cref{rm:well-sepclique-basd}(a), we could construct a balanced separator of size $O(k\sqrt{n})$ as each clique has at most $k$ vertices, and the clique-based separator has $\sqrt{n}$ cliques. Thus, using standard algorithms~\cite{Frederickson1987,Wulff-Nilsen11}, $G$ admits an $r$-division $\mathcal{R}$ such that:
\begin{enumerate}
    \item $\mathcal{R}$ has $O(kn/\sqrt{r})$ clusters, each induced a connected subgraph of $G$ of size at most $r$. 
    \item Each region $R\in \mathcal{R}$ has at most $O(kr)$ vertices having edges outside $R$, called boundary of $R$, and denoted by $\bdry R$. 
    \item $\sum_{R\in \mathcal{R}} |\bdry R|= O(kn/\sqrt{r})$. That is, the total number of boundary vertices counted with multiplicity is $O(kn/\sqrt{r})$.
    \item  Every vertex of $G$ is in a region in $\mathcal{R}$.
\end{enumerate}

By \Cref{rm:well-sepclique-basd}(b), we can construct an implicit representation of a clique-based separator for $G$ in $\Tilde{O}(n)$ time. As each clique has size at most $k$, from the implicit representation, we could obtain all the vertices in the balanced separator in time $\Tilde{O}(nk)$. Thus, following standard techniques~\cite{Frederickson1987}, we can construct an $r$-division in $\Tilde{O}(nk)$ time.

Next, we use distance encoding following the approximate diameter algorithm in \Cref{sec:diameter-UDG}. 
Here, the difference is that we no longer need to choose a representative per clique; we have all the boundary vertices and the number of boundary vertex per cluster $R\in \mathcal{R}$ is at most $O(kr)$. For each cluster $R$, we take all vertices in $\bdry R$ to construct the sequence of vertices in the pattern construction. Therefore, the distance  $\hat{d}_G(u,v)$ as defined in \Cref{eq:approx-dist} is the exact distance between $d_G(u,v)$. Now we apply the same algorithm in \Cref{sec:diameter-UDG} with the following modifications:

\begin{itemize}
    \item The set of cliques $\mathcal{C}$ being boundary vertices $\bigcup_{R\in \mathcal{R}} \! \bdry R$; each boundary vertex is a singleton~clique.
    \item  In \Cref{eq:delta-u-C}, we do not add $1$. Specifically, we set $\Delta(u,\mathcal{C}) = \max_{x\in C}d_G(u,x)$.
\end{itemize}

Then, we get an algorithm for computing the \emph{exact eccentricities} of every vertex. 
Thus, the returned diameter is an exact diameter. Next, we analyze the running time.

\paragraph{Running time.~} Note that SSSP in similar-size pseudo-disk graphs with constant complexity could be computed in $\Tilde{O}(n)$ by \Cref{thm:SSSP} in \Cref{subset:SSP}. As $\sum_{R\in \mathcal{R}} |\bdry R|= O(kn/\sqrt{r})$ the running time of Step 1 is $\Tilde{O}(kn^2/\sqrt{r})$.  The number of patterns is $(kr^2)^d = k^d r^{2d}$. Thus, Step 2 could be implemented in $O((nk/\sqrt{r})\cdot (kr^2)^d \cdot r) = nk^{d+1}r^{(4d+1)/2}$. The running time of the last step is $\Tilde{O}(kn^2/\sqrt{r} + knr^{3/2})$. Thus, the total running time of the algorithm is:
\begin{equation*}
    \Tilde{O}\left(kn^2/\sqrt{r} + nk^{d+1}r^{(4d+1)/2}\right) =  \Tilde{O}(k^{11/9}n^{2-1/18})
\end{equation*}
for $r = (n/k^d)^{\frac{1}{2d+1}}$ and $d = 4$.

\begin{remark} For unit disk graphs, one can obtain a better dependency on $k$ by using the separator by~\cite{MTTV97} that has size $O(\sqrt{k n})$, which then implies a truly subquadratic time algorithm for computing exact diameter in unit disk graphs for a slightly larger value of $k$. 
\end{remark}

\subsection{Single-source shortest paths in pseudo-disk graphs}\label{subset:SSP}

We study the problem of single-source shortest paths in a pseudo-disk graph.  For an unweighted unit-disk graph, computing the single-source shortest paths (SSSP) can be done in time $O(n\log n)$, with $n$ as the number of vertices in the graph. 
There are a number of algorithms~\cite{Efrat2001-hm,Cabello2015-vo,Chan2016-sy}  with this running time as reported in the literature.  Notice that the running time is tight \cite{Cabello2015-vo} --- reduction from the problem of finding the \emph{maximum gap} in a set of numbers shows that deciding if the unit disk graph is connected requires $\Omega(n\log n)$ time. 
In the following we adapt the $O(n\log n)$ algorithm by Chan and Skrepetos~\cite{Chan2016-sy} to the setting of pseudo-disks.

The main idea is to implement the breadth-first search without explicitly constructing the entire graph. 
The algorithm starts from the source $s$ and proceeds in $n-1$ steps. 
In step $i$, suppose we have already found all pseudo-disks $S_{i-1}$ within distance exactly $i-1$ from the source $s$ (the frontier) and now finds the pseudo-disks of distance $i$ from $s$, i.e., the disks that can be reached from the disks $S_{i-1}$ and are not yet found in earlier steps (i.e., in $\bigcup_{j<i} S_j$). 
To aid the steps, we put a grid of side length $r\sqrt{2}$ and bucket the centers of the pseudo-disks in these grid cells. 
Any two pseudo-disks with centers in the same cell intersect with each other for sure. 
Therefore, once we identify a pseudo-disk $C_p$ in $S_{i-1}$ with center $p$ in a grid cell, all pseudo-disks with centers in the same grid cell will be included in $S_{i}$ if they are not already included. 
Furthermore, for any two grid cells with distance at least $2R$ from each other, two pseudo-disks centered at the two cells respectively do not intersect. Thus for each cell $c$ touched by $S_{i-1}$, we check at most $O(R^2/r^2)=O(1)$ cells (called the \emph{neighboring cells} of $c$) potential pseudo-disks to be included in $S_{i}$.

\paragraph{Red-blue intersection.}
To efficiently find $S_{i}$ from $S_{i-1}$, we use the \emph{red-blue intersection algorithm} to identify, for a set of pseudo-disks in $S_{i-1}$ with center in one cell $c$ (denoted by $S_{i-1}(c)$), the pseudo-disks in another cell $c'$ that intersect at least one pseudo-disk in $S_{i-1}(c)$. 
The name is justified if we color the pseudo-disks in $S_{i-1}(c)$
\emph{red} and the pseudo-disks in cell $c'$ \emph{blue}. 
The centers of the pseudo-disks in the two sets are separated by a line that separates $c$ and $c'$.
The following algorithm only needs the assumption that the boundary of each pseudo-disk is defined by a constant number of algebraic arcs.  \begin{definition}[Red-blue intersection problem]
Given a set of $n_r$ red pseudo-disks with centers below a horizontal line $h$ and another set of $n_b$ blue pseudo-disks with centers above $h$, determine for each blue pseudo-disk whether there is a red pseudo-disk that intersects with it. 
\end{definition}

We adapt the algorithm in~\cite[Subproblem~2]{Chan2016-sy} to accommodate pseudo-disks. 
We first compute the upper envelope of the red pseudo-disks $U_r$ and then run a sweeping line algorithm to check for each blue pseudo-disk whether any part of it is below $U_r$.

To compute the upper envelope $U_r$, we take each pseudo-disk $C_p$ and consider part of the boundary $C_p^+$ that is above $h$. 
For any two pseudo-disks $C_p, C_q$, $C_p^+$ and $C_q^+$ intersect at most twice.  
Now we compute the upper envelope of the segments $\{C_p^+\}$ for all red pseudo-disks $\{C_p\}$.
We call a collection of curves \EMPH{$s$-intersecting} if any pair in a set of curves (or curve segments) only intersects at most $s$ times.  The upper envelope of $n$ curve segments that are $s$-intersecting has complexity $\lambda_{s+2}(n)$, where 
$\lambda_s(n)$ is the maximum length of an $(n, s)$-Davenport-Schinzel sequence%
\footnote{An \EMPH{$(n, s)$-Davenport-Schinzel sequence} is a sequence of $n$ symbols such that no two adjacent symbols are the same and there is no subsequence of any alternation of length $s+2$ with two distinct symbols. It is known~\cite{Toth2017-du,Pettie2015-pm} that $\lambda_2(n)=2n-1$, $\lambda_3(n)=2n\alpha(n)+O(n)$, and $\lambda_4(n)=O(n2^{\alpha(n)})$, where $\alpha(n)$ is the inverse Ackermann function.}. 
Computing the upper envelope of $n$ curve segments where each pair intersects at most $s$ times can be done in time $O(\lambda_{s+1}(n)\log n)$~\cite{Hershberger1989-zu}. For our case, $s=2$, thus we can find $U_r$ in time $O(n_r\alpha(n_r)\log n_r)$ and the complexity of $U_r$ is $O(n_r 2^{\alpha(n_r)})$.

Once we have the upper envelope $U_r$ of the red pseudo-disks, we will run a sweeping line algorithm and check for each blue pseudo-disk whether any part of it is below $U_r$. After we sort all the boundary vertices of the blue pseudo-disks, the scan can be done in time linear to the complexity of $U_r$ and $n_b$, since a vertical sweeping line only intersects a pseudo-disk at an interval. 

In summary, we have the following lemma. 

\begin{lemma}\label{lm:red-blue-gen} 
In time $O(n_b\log n_b +n_r\alpha(n_r)\log n_r+n_r 2^{\alpha(n_r)})$, we can solve the red-blue intersection problem of $n_r$ pseudo-disks and $n_b$ blue pseudo-disks.
\end{lemma}

Now we can conclude the SSSP algorithm for fat pseudo-disks of comparable size.

\begin{theorem}\label{thm:SSSP}
For $n$ fat pseudo-disks of bounded size, we can solve the single-source shortest paths problem in time $O\Paren{n\cdot(2^{\alpha(n)}+\alpha(n)\log n)}$, where $\alpha(n)$ is the inverse Ackermann function.
\end{theorem}
\begin{proof}
For step $i$ of the BFS algorithm, by induction we have the pseudo-disks $S_{i-1}$ at exactly distance $i-1$ from source $s$. 
For each cell $c$ with at least one pseudo-disk in $S_{i-1}$, we consider all cells $c'$ of distance within $2R$ from $c$ and run the red-blue intersection algorithm to find disks that intersect with at least one disk of $c$ in $S_{i-1}$. 
We filter out disks that are already discovered and arrive at $S_i$. Each cell $c$ that is non-empty (containing at least one center of the pseudo-disks) is only visited at most a constant number of times, either when one of the pseudo-disks centered inside $c$ enters the frontier or when at least one of the pseudo-disks centered at a neighboring cell of $c$ enters the frontier. 
The total running time is $O\Paren{n\cdot (2^{\alpha(n)}+\alpha(n)\log n)}$.
\end{proof}

\section{Computing Diameter+1}\label{sec:plusone}

In this section we present a truly-subquadratic time algorithm for +1-approximation of the graph diameter for unit-disk graphs and pseudo-disk graphs. The current algorithm achieves a +2 approximation and the main reason is that we take a representative vertex from each clique in the clique-based $r$-clustering. 
Suppose the shortest path between two vertices $u$ and $v$, with $v$ inside a cluster $R$ and $u$ outside $R$, goes through a vertex $x$ in the clique $C$ on the boundary $\partial\! R$. The earlier algorithm considers the shortest path from $u$ to the representative vertex $s$ in the same boundary clique $C$ containing $x$. Clearly $d(u, s)\leq d(u, x)+1$ and $d(v, s)\leq d(v, x)+1$. Using distance $d(u, s)+d(v, s)$ gives possibly an overestimate of $+2$ error for $d(u, v)=d(u, x)+d(v, x)$. This is the argument in \Cref{lm:dist-via-pattern}.

To fix this issue, we instead choose a dummy vertex $s$ for clique $C$ (rather than an actual vertex in the clique $C$). We add dummy edges from $s$ to all vertices in $C$ with length of $1/2$. Essentially one can replace the clique $C$ by a star of the dummy edges --- the distance between any pair of vertices in $C$ is still~$1$. This operation does not change the shortest path length for any pair of input vertices $(u,v)$. 
If the shortest path $P(u, v)$ has zero or one vertex in $C$, $P(u, v)$ remains the same. If $P(u, v)$ has two vertices $x, y$ in $C$, then the edge $xy$ is replaced by two edges $xs$ and $sy$ with total length unchanged. $P(u, v)$ cannot have three or more vertices of $C$ since $P$ is the shortest.

Next we will calculate the shortest path of $d(u, s)$ and $d(v, s)$ and use $d(u, s)+d(v, s)$ as an estimate of $d(u, v)$. 
There are two scenarios. In the first scenario, the shortest path between $u$ and $v$ goes through one vertex $x$ in the clique $C$.
See~\Cref{fig:path-clique} (Left). In this case, $d(u, s)\leq d(u, x)+1/2$ and $d(v, s)\leq d(v, s) +1/2$ and thus we have an additive error of $+1$. In the second case, the shortest path between $u$ and $v$ goes through two vertices $x, y$ in the clique $C$.
See~\Cref{fig:path-clique} (Right). In this case $d(u, v)=d(u, x)+1+d(y, v)=d(u, s)+d(v,s)$. Thus the distance $d(u, s)+d(v, s)$ is accurate for $d(u, v)$. We use the same idea of representing the distance $d(u, s)+d(v, s)$ by the sum of $d(u, s_0)$---with only one dummy node $s_0$ for each cluster $R$ in the $r$-clustering---and $d(\patt_u, v)$, which can be obtained by computing the distance of $v\in R$ and all possible patterns $\patt$. By \Cref{thm:dummy-graph-VC}, the dummy graph (obtained by replacing every clique with a star) has VC-dimension of pattern vectors bounded by 4 and hence, the number of patters remains unchanged. The only ingredient we need for our main results is computing SSSP from dummy nodes.

\begin{figure}
\centering
\includegraphics[width=0.7\linewidth]{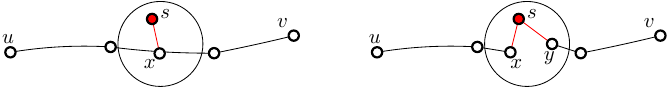}
\caption{Using shortest path from vertices to dummy representative vertices of cliques in the clique-based clustering gives an additive error of $+1$.}
\label{fig:path-clique}
\end{figure}

\paragraph{SSSP from dummy nodes.}
First we need to calculate the distance $d(u, s)$ for all dummy vertices $s$ of cliques on $\partial\! R$ with all vertices $u$ in $R$. In fact, we will run a SSSP from each of the dummy nodes for cliques on the clique-based $r$-clustering. This can be implemented by running a multi-source shortest path tree, starting from the real vertices of a clique $C$. At the end, we add a distance $1/2$ to all distances to obtain the distance from a dummy node $s$ to all vertices. Since the vertices in a clique stay in at most four unit-sized grid cells, we use the same technique of red-blue intersection in \Cref{subset:SSP}. Thus this can be done in time $O(n\log n)$ for unit-disk graph and in time $\tilde{O}(n)$ for pseudo-disks.

\section{Open Problems}

In this paper, we presented truly subquadratic algorithms for a +1-approximation to the graph \textsc{Diameter} in an unweighted unit-disk graph. 
The obvious open question is whether this can be done for the \emph{exact} diameter, thus resolving the long-standing open question. 
Another open problem is whether the results can be extended to the intersection graph of disks of possibly different radii. 
The challenge there is to develop something similar to an $r$-division. 
While the clique-based separator by de Berg~\cite{deBerg23} works for general disk graphs, there are challenges applying the separator (or some other variants) recursively to get a nice subdivision with bounded boundary size per piece.

\section*{Acknowledgements} 
The authors would like to thank Mark de Berg, Arnold Filtser and Da Wei (David) Zheng for useful discussion at CGWeek of 2024 about the improvement from $+2$ to $+1$-approximation. 
}

\newpage
\small
\bibliographystyle{alphaurl}
\bibliography{VC-dim}

\end{document}